\documentclass[11pt,a4paper,onecolumn,accepted=2024-07-10]{quantumarticle}

\pdfoutput=1

\usepackage{geometry}
\usepackage[utf8]{inputenc}
\usepackage{lipsum}

\usepackage{amsmath, amssymb,amsfonts,amsthm,mathtools}
\usepackage{xspace,graphicx,relsize,bm,xcolor}
\usepackage{soul} % provides  s p a c i n g o u t, underlining and some derivatives such as overstriking and highlighting: https://ctan.mc1.root.project-creative.net/macros/generic/soul/soul.pdf

\usepackage{parskip}  % no indentation but some space between paragraphs

\usepackage[
    backend=biber,
    style=alphabetic,
    sorting=anyt,
    minalphanames=3,
    maxalphanames=3,
    maxnames=99,
    backref=true
    ]{biblatex}
\DefineBibliographyStrings{english}{%
  backrefpage = {page},% originally "cited on page"
  backrefpages = {pages},% originally "cited on pages"
}

\usepackage{sepfootnotes}
\newendnotes{x}
\renewcommand\xnotesize\normalsize

\usepackage{hyperref}
\usepackage[margin=1.75cm]{geometry}
\definecolor{linkcol}{rgb}{0.0,0.55,0.7}
\definecolor{citecol}{rgb}{0.0, 0.6, 0.45}
\definecolor{urlcol}{rgb}{0.7, 0.0, 0.55}
\hypersetup{
	colorlinks,
	linkcolor={linkcol},
	citecolor={citecol},
	urlcolor={urlcol}
}

\usepackage{url}
\usepackage{mleftright}
\usepackage{hyperref}
\usepackage{multirow}
\usepackage{physics}  % all braket notation at once

\usepackage{algorithm}
\usepackage{algpseudocode}

\usepackage{cleveref}

\usepackage{authblk}  % for Authors

\def\01{\{0,1\}}

\DeclareMathOperator*{\argmin}{arg\,min}

\newcommand{\rsubseteq}{\rotatebox[origin=c]{270}{$\subseteq$}}

\usepackage{dsfont}
\newcommand{\doubleyou}{\mathrm{w}}
\newtheoremstyle{mydefinitionsty}% 〈name〉
{10pt}% 〈Space above〉
{10pt}% 〈Space below〉
{}% 〈Body font〉
{}% 〈Indent amount〉1
{}% 〈Theorem head font〉
{}% 〈Punctuation after theorem head〉
{.5em}% 〈Space after theorem head〉2
{\textbf{\thmname{#1}~\thmnumber{#2}:  }\thmnote{(#3)}}% 〈Theorem head spec (can be left empty, meaning ‘normal’)〉
\theoremstyle{mydefinitionsty}
\newtheorem{definition}{Definition}

\newtheorem{example}{Example}

\newtheorem{observation}{Observation}

\newtheorem{simplealgorithm}{Algorithm}

\newtheoremstyle{myproblemsty}% 〈name〉
{10pt}% 〈Space above〉
{10pt}% 〈Space below〉
{}% 〈Body font〉
{}% 〈Indent amount〉1
{}% 〈Theorem head font〉
{}% 〈Punctuation after theorem head〉
{.5em}% 〈Space after theorem head〉2
{\textbf{\thmname{#1}~\thmnumber{#2}:  }\thmnote{(#3)}\newline}% 〈Theorem head spec (can be left empty, meaning ‘normal’)〉
\theoremstyle{myproblemsty}

\newtheoremstyle{mythmsty}% 〈name〉
{10pt}% 〈Space above〉
{10pt}% 〈Space below〉
{\itshape}% 〈Body font〉
{}% 〈Indent amount〉1
{}% 〈Theorem head font〉
{}% 〈Punctuation after theorem head〉
{.5em}% 〈Space after theorem head〉2
{\textbf{\thmname{#1}~\thmnumber{#2}:  }\thmnote{(#3)}}% 〈Theorem head spec (can be left empty, meaning ‘normal’)〉
\theoremstyle{mythmsty}

\newtheorem{theorem}{Theorem}
\newtheorem{lemma}{Lemma}

\newtheorem{question}{Question}

\definecolor{alexcolor}{rgb}{0.0, 0.47, 0.75}   % {0.27, 0.51, 0.71}  %

\definecolor{questioncolor}{rgb}{0.36, 0.54, 0.66}

\title{Potential and limitations of random Fourier features for dequantizing quantum machine learning}

\begin{document}

\author[1]{Ryan Sweke}
\author[2,3]{Erik Recio-Armengol}
\author[4]{Sofiene Jerbi}
\author[4,5]{Elies Gil-Fuster}
\author[7]{Bryce Fuller}
\author[4,5,6]{Jens Eisert}
\author[4]{Johannes~Jakob~Meyer}

\affil[1]{IBM Quantum, Almaden Research Center, San Jose, CA, USA}
\affil[2]{ICFO-Institut de Ciencies Fotoniques, The Barcelona Institute of Science and Technology, 08860 Castelldefels, Spain}
\affil[3]{Eurecat, Centre Tecnologic de Catalunya, Multimedia Technologies, Barcelona, Spain}
\affil[4]{Dahlem Center for Complex Quantum Systems, Freie Universit\"at Berlin, Berlin, Germany}
\affil[5]{Fraunhofer Heinrich Hertz Institute, 10587 Berlin, Germany}
\affil[6]{Helmholtz-Zentrum Berlin f\"{u}r Materialien und Energie, 14109 Berlin, Germany}
\affil[7]{IBM Quantum, IBM T.J. Watson Research Center, Yorktown Heights, NY 10598}

\begin{abstract}
\noindent Quantum machine learning is arguably one of the most explored applications of near-term quantum devices. Much focus has been put on notions of variational quantum machine learning where \emph{parameterized quantum circuits} (PQCs) are used as learning models. These PQC models have a rich structure which suggests that they might be amenable to efficient dequantization via \emph{random Fourier features} (RFF). In this work, we establish necessary and sufficient conditions under which RFF does indeed provide an efficient dequantization of  variational quantum machine learning for regression. We build on these insights to make concrete suggestions for PQC architecture design, and to identify structures which are necessary for a regression problem to admit a potential quantum advantage via PQC based optimization.
\end{abstract}

\maketitle

\section{Introduction}\label{s:introduction}

In recent years, the technique of using \emph{parameterized quantum circuits} (PQCs) to define a model class, which is then optimized over via a classical optimizer, has emerged as one of the primary methods of using near-term quantum devices for machine learning tasks~\cite{Cerezo_2021,Benedetti_2019}. We will refer to this approach as \textit{variational quantum machine learning} (variational QML), although it is often also referred to as hybrid quantum/classical optimization. While a large amount of effort has been invested in both understanding the theoretical properties of variational QML, and experimenting on trial datasets, it remains unclear whether variational QML on near-term quantum devices can offer any meaningful advantages over state-of-the-art classical methods.

One approach to answering this question is via \textit{dequantization}. In this context, the idea is to use insights into the structure of PQCs, and the model classes that they define, to design quantum-inspired classical methods which can be proven to match the performance of variational QML. Ultimately, the goal is to understand when and why variational QML can be dequantized, in order to better identify the PQC architectures, optimization algorithms and problem types for which one might obtain a meaningful quantum advantage via variational QML. 

In order to discuss notions of dequantization of variational QML, we note that for typical applications variational QML consists of two distinct phases. Namely, a \textit{training} stage and an \textit{inference} stage. In the training stage,  one uses the available training data to identify an optimal PQC model, and in the inference stage one uses the identified model to make predictions on previously unseen data, or in the case of generative modelling, to generate new samples from the unknown data distribution.

A variety of works have recently proposed dequantization methods for \textit{inference} with PQC models. The first such work was Ref.~\cite{classical_surrogates}, which used insights into the functional analytic structure of PQC model classes to show that, given a trained quantum model, one can sometimes efficiently extract a purely classical model -- referred to as a \textit{classical surrogate} -- which performs inference just as well as the PQC model. More recently, Ref.~\cite{jerbi2023shadows} used insights from shadow tomography to show that it is sometimes possible to extract a classical shadow of a trained PQC model -- referred to as a \textit{shadow model} -- which is again guaranteed to perform as well as the PQC model for inference. Interestingly, however, Ref.~\cite{jerbi2023shadows} also proved that, under reasonable complexity-theoretic assumptions, there exist PQC models whose inference \textit{cannot} be dequantized by any efficient method -- i.e., \textit{all} efficient methods for the dequantization of PQC inference possess some fundamental limitations.

In this work, we are concerned with dequantization of the \textit{training} stage of variational QML -- i.e., the construction of efficient classical learning algorithms which can be proven to match the performance of variational QML in learning from data. To this end, we start by noting that Ref.~\cite{jerbi2023shadows} constructed a learning problem which admits an efficient variational QML algorithm but, again under complexity-theoretic assumptions, cannot be dequantized by any efficient classical learning algorithm. As such, we know that all methods for the dequantization of variational QML must posses some fundamental limitations, and for any given method we would like to understand its domain of applicability. 

With this in mind, one natural idea is to ask whether the effect of \textit{noise} allows direct efficient classical simulation of the the PQC model, and, therefore, of the entire PQC training process and subsequent inference. Indeed, a series of recent works has begun to address this question, and to delineate the conditions under which noise renders PQC models classically simulatable~\cite{fontana2022spectral,fontana2023classical,shao2023simulating}. Another recent result has shown that the presence of \textit{symmetries}, often introduced to improve PQC model performance for symmetric problems~\cite{BerlinSymmetry,larocca2022group}, can also constrain PQC models in a way which allows for efficient classical simulation~\cite{SymmetryDequantization}. 

In this work however, inspired by the dequantization of inference with PQC models, we start from the idea of ``training the surrogate model''.  More specifically, Ref.~\cite{classical_surrogates} used the insight that the class of functions realizable by PQC models is a subset of trigonometric polynomials with a specific set of frequencies~\cite{data_encoding}, to ``match'' the trained PQC model to the closest trigonometric polynomial with the correct frequency set (which is then the classical surrogate for inference). This, however, immediately suggests the following approach to dequantizing the \textit{training} stage of variational QML --  simply directly optimize from data over the ``PQC-inspired'' model class of trigonometric polynomials with the appropriate frequency set. Indeed, this is in some sense what happens during variational QML! 

The above idea was explored numerically in the original work on classical surrogates for inference~\cite{classical_surrogates}. Unfortunately, for typical PQC architectures the number of frequencies in the corresponding frequency set grows exponentially with the problem size (defined as the dimensionality of the input data), which prohibits efficient classical optimization over the relevant class of trigonometric polynomials. However, for the PQC architectures for which numerical experiments were possible, direct classical optimization over the PQC-inspired classical model class yielded trained models which outperformed those obtained from variational QML. 

\newpage
Inspired by Ref.~\cite{classical_surrogates}, the subsequent work of Ref.~\cite{RFF} introduced a method for addressing the efficiency bottleneck associated with exponentially growing frequency sets. The authors of Ref.~\cite{RFF} noticed that all PQC models are linear models with respect to a  trigonometric polynomial feature map. As such, one can optimize over all PQC-inspired models via kernel ridge regression, 
which will be efficient if one can efficiently evaluate the kernel defined from the feature map. While naively evaluating the appropriate kernel classically will be inefficient -- again due to the exponential growth in the size of the frequency set -- the clever insight of Ref.~\cite{RFF} was to see that one can gain efficiency improvements by using the technique of \textit{random Fourier features}~\cite{rahimi2007random} to \textit{approximate} the PQC-inspired kernel. Using this technique, Ref.~\cite{RFF} obtained a variety of theoretical results concerning the sample complexity required for RFF-based regression with the PQC-inspired kernel to yield a model whose performance matches that of variational QML. However, the analysis of Ref.~\cite{RFF} applied only to PQC architectures with \textit{universal} parameterized circuit blocks -- i.e.,  PQC models which can realize \textit{any} trigonometric polynomial with the appropriate frequencies. This contrasts with the PQC models arising from practically relevant PQC architectures, which due to depth constraints, can only realize a subset of trigonometric polynomials. 

In light of the above, the idea of this work is to further explore the potential and limitations of RFF-based linear regression as a method for dequantizing the training stage of variational QML, with the goal of providing an analysis which is applicable to practically relevant PQC architectures. In particular, we identify a collection of necessary and sufficient requirements -- of the PQC architecture, the regression problem, and the RFF procedure -- for RFF-based linear regression to provide an efficient classical dequantization method for PQC based regression. This allows us to show clearly that:

\begin{enumerate}
\item RFF-based linear regression \textit{cannot} be a generic dequantization technique. At least, there exist regression problems and PQC architectures for which  RFF-based linear regression \textit{cannot} efficiently dequantize PQC based regression. As mentioned before, we already knew from the results of Ref.~\cite{jerbi2023shadows} that \textit{all} dequantization techniques must posses some limitations, and our results shed light on the specific nature of these limitations for RFF-based dequantization.
\item There exist practically relevant PQCs, and regression problems, for which RFF-based linear regression can be \textit{guaranteed} to efficiently produce output models which perform as well as the best possible output of PQC based optimization. In other words, there exist problems and PQC architectures for which PQC dequantization via RFF is indeed possible.
\end{enumerate}
Additionally, using the necessary and sufficient criteria that we identify, we are able to provide concrete recommendations for PQC architecture design, in order to mitigate the possibility of dequantization via RFF. Moreover, we are able to identify a necessary condition on the structure of a regression problem which ensures that dequantization via RFF is \textit{not} possible. This, therefore, provides a guideline for the identification of problems which admit a potential quantum advantage (or at least, cannot be dequantized via RFF-based linear regression). 

This paper is structured as follows: We begin in Section~\ref{s:setting_prelim} by providing all the necessary preliminaries and background material. Following this, we proceed in Section~\ref{s:PQCvsRFF} to motivate and present RFF-based linear regression with PQC-inspired kernels as a method for the dequantization of variational QML. Given this, we then go on in Section~\ref{ss:RFF_analysis} to provide a detailed theoretical analysis of RFF-based linear regression with PQC-inspired kernels. Finally, we conclude in Section~\ref{s:conclusions} with a discussion of the consequences of the previous analysis, and with an overview of natural directions for future research.

\section{Setting and preliminaries}\label{s:setting_prelim}

Here we provide the setting and required background material.

\subsection{Statistical learning framework}\label{ss:stat_framework}

Let $\mathcal{X}$ denote a set of all possible data points, and $\mathcal{Y}$ a set of all possible labels. In this work, we will set $\mathcal{X} \coloneqq [0,2\pi)^d\subset\mathbb{R}^d$ 
for some integer $d$ and $\mathcal{Y} = \mathbb{R}$.  We assume the existence of some unknown probability distribution $P$ over $\mathcal{X}\times\mathcal{Y}$, which we refer to as a regression problem. Note that we consider $d$, the dimensionality of the input data, as the size of the problem, and as such the relevant scaling parameter for the analysis of algorithms which aim to solve the problem. Additionally, we assume a parameterized class of functions $\mathcal{F} = \{f_{\theta}\colon\mathcal{X}\rightarrow\mathcal{Y}\,|\,\theta\in\Theta\}$, which we call hypotheses. Given access to some finite dataset $S = \{(x_i,y_i)\sim P\}|_{i = 1}^n$ the goal of the regression problem specified by $P$ is to identify the optimal hypothesis $f_{\theta^*}$, i.e.,  the hypothesis which minimizes the \textit{true risk}, defined via
\begin{equation}
R(f) \coloneqq \mathop{\mathbb{E}}_{(x,y)\sim P}\left[\mathcal{L}(y,f(x))\right],
\end{equation}
where $\mathcal{L}\colon\mathcal{Y}\times\mathcal{Y}\rightarrow\mathbb{R}$ is some loss function. In this work, we will consider only the quadratic loss defined via 
\begin{equation}
{\mathcal{L}(y,y')
\coloneqq (y-y')^2}.
\end{equation} 
We also define the \textit{empirical risk} with respect to the dataset $S$ as 
\begin{equation}
\hat{R}(f) \coloneqq \frac{1}{n}\sum_{i = 1}^n\mathcal{L}(y_i,f(x_i)).
\end{equation}

\subsection{Linear and kernel ridge regression}

Linear ridge regression  and kernel ridge regression are two popular classical learning algorithms. In linear ridge regression we consider \textit{linear} functions $f_{\doubleyou}(x) = \langle \doubleyou,x\rangle$, and given a dataset $S=\{(x_i,y_i)\}|^n_{i=1}$, we proceed by minimizing the empirical risk, regularized via the $2$-norm, i.e., 
\begin{equation}\label{eq:lrr_reg_risk}
\hat{R}_\lambda(f_{\doubleyou}) \coloneqq \frac{1}{n}\sum_{i=1}^n\left(y_i - \langle \doubleyou,x_i\rangle \right)^2 + \lambda \|\doubleyou\|_2^2.
\end{equation}
With this regularization, which is added to prevent over-fitting,  minimizing Eq.~\eqref{eq:lrr_reg_risk} becomes a convex quadratic problem, which admits the closed form solution
\begin{equation}
\doubleyou = \left(\hat{X}^T\hat{X} + \lambda n\mathds{1}\right)^{-1}\hat{X}^T\hat{Y},
\end{equation}
where $\hat{X}$ is the $n\times d$ ``data matrix'' with $x_i$ as rows, and $\hat{Y}$ is the $n$ dimensional ``target vector'' with $y_i$ as the $i$'th component~\cite{mohri2018foundations}. Linear ridge regression requires $\mathcal{O}(nd)$ space and $\mathcal{O}(nd^2 + d^3)$ time. As linear functions are often not sufficiently expressive, a natural approach is to consider linear functions in some higher dimensional feature space.  More specifically, one assumes a feature map $\phi\colon\mathbb{R}^d\rightarrow \mathbb{R}^{D}$, and then considers linear functions of the form $f_v(x) = \langle v,\phi(x)\rangle$, where $v$ is an element of the feature space $\mathbb{R}^{D}$. Naively, one could do linear regression at a space and time cost of $\mathcal{O}(nD)$ and $\mathcal{O}(nD^2 + D^3)$,  respectively. However, often we would like to consider $D$ extremely large (or infinite) and this is therefore infeasible. The solution is to use ``the kernel trick'' and consider instead a kernel function $K\colon\mathcal{X}\times\mathcal{X}\rightarrow \mathbb{R}$ which satisfies
\begin{equation}
K(x,x') = \langle\phi(x),\phi(x')\rangle,
\end{equation}
but which can ideally be evaluated more efficiently than by explicitly constructing $\phi(x)$ and $\phi(x')$ and taking the inner product.
Given such a function, we know that the minimizer of the regularized empirical risk is given by 
\begin{align}
f_\alpha(x) &= \sum_{i = 1}^n\alpha_iK(x_i,x)\\
\nonumber
            &= \left\langle \sum_{i = 1}^n\alpha_i \phi(x_i),\phi(x)\right\rangle\\
            \nonumber
            & = \langle v,\phi(x)\rangle,
            \nonumber
\end{align}
where 
\begin{equation}\label{eq:KRR}
\alpha = \left(\hat{K} + n\lambda\mathds{1}\right)^{-1}\hat{Y},
\end{equation}
with $\hat{K}$ the kernel matrix (or Gram matrix) with entries $\hat{K}_{i,j} = K(x_i,x_j)$. Solving Eq.~\eqref{eq:KRR} is known as \emph{kernel ridge regression}. If one assumes that evaluating $K(x,x')$ requires constant time, then kernel ridge regression has space and time cost $\mathcal{O}(n^2)$ and $\mathcal{O}(n^3)$, respectively.
We note that in practice one hardly ever specifies the feature map $\phi$, and instead works directly with a suitable kernel function $K$.

\subsection{Random Fourier features}\label{ss:RFF}

For many applications, in which the number of samples $n$ can be extremely large, a space and time cost of $\mathcal{O}(n^2)$ and $\mathcal{O}(n^3)$, respectively, prohibits the implementation of kernel ridge regression. This has motivated the development of methods which can bypass these complexity bottlenecks. The method of \emph{random Fourier features} (RFF) is one such method~\cite{rahimi2007random}. To illustrate this method we follow the presentation of Ref.~\cite{genRFF}, and start by assuming that the kernel $K\colon\mathcal{X}\times\mathcal{X}\rightarrow \mathbb{R}$, has \textit{an integral representation}. More specifically, we assume there exists some probability space $(\Phi,\pi)$ and some function $\psi\colon\mathcal{X}\times \Phi \rightarrow \mathbb{R}$ such that for all $x,x'\in \mathcal{X}$ one has that
\begin{equation}\label{eq:integral_rep}
K(x,x') = \int_{\Phi}\psi(x,\nu)\psi(x',\nu)\,\mathrm{d}\pi(\nu).
\end{equation}
The method of random Fourier features is then based around the idea of approximating $K(x,x')$ by using Monte-Carlo type integration to perform the integral in Eq.~\eqref{eq:integral_rep}. More specifically, one uses
\begin{equation}\label{eq:RFF_approx}
K(x,x')\approx  \langle\tilde{\phi}_M(x), \tilde{\phi}_M(x')\rangle,
\end{equation}
where $\tilde{\phi}_M\colon\mathcal{X}\rightarrow \mathbb{R}^M$ is a randomized feature map of the form
\begin{equation}\label{eq:RFF_feature_map}
\tilde{\phi}_M(x) = \frac{1}{\sqrt{M}}\big(\psi(x,\nu_1),\ldots,\psi(x,\nu_M)\big),
\end{equation}
where $\nu_1,\ldots,\nu_m$ are $M$ features sampled randomly from $\pi$. Using the approximation in Eq.~\eqref{eq:RFF_approx} allows one to replace kernel ridge regression via $K$ with linear regression with respect to the random feature map $\tilde{\phi}_M$. This yields a learning algorithm with time and space cost $\mathcal{O}(nM)$ and $\mathcal{O}(nM^2 + M^3)$, respectively, which is more efficient than kernel ridge regression whenever $M<n$. Naturally, the quality (i.e.,  true risk) of the output solution will depend heavily on how large $M$ is chosen. However, we postpone until later a detailed discussion of this issue, which is central to the results and observations of this work.

So far we have \textit{assumed} the existence of an integral representation for the kernel, as in Eq.~\eqref{eq:integral_rep}. However, such a representation is not typically provided, and as such the first step towards the implementation of linear regression with RFF is the derivation of an integral representation for the kernel of interest. Luckily however, for \textit{shift-invariant} kernels, which are kernels of the form $K(x,x') = \overline{K}(x-x')$ for some function $\overline{K}\colon\mathcal{X}\rightarrow \mathbb{R}$, the integral representation can be easily derived from the Fourier transform of $\overline{K}$~\cite{rahimi2007random,errorRFF}. More specifically, for any shift-invariant kernel we can write
\begin{equation}
\overline{K}(x-x') = \int_{\omega\in\mathcal{X}}e^{i\langle\omega,x-x'\rangle}q(\omega)\mathrm{d}\omega,
\end{equation}
where $q:\mathcal{X}\rightarrow \mathbb{R}$ is the Fourier transform of $\overline{K}$. Bochner's theorem ensures that $q$ is a non-negative measure, and when the kernel is scaled such that $\overline{K}(0) = 1$, then it additionally ensures that $q$ is indeed a proper probability distribution. Additionally, as the kernel is real-valued, we can replace the integrand  $e^{i\langle\omega,x-x'\rangle}$ with $\cos(\langle\omega,x-x'\rangle)$. Doing this, we then have
\begin{align}
K(x-x') &= \int_{\omega\in\mathcal{X}}\cos(\langle\omega,x-x'\rangle)q(\omega)\mathrm{d}\omega\nonumber
\\
&=\frac{1}{2\pi}\int_{\omega\in\mathcal{X}}\int_{\gamma\in[0,2\pi)} \sqrt{2}\cos(\langle\omega, x\rangle + \gamma)\sqrt{2}\cos(\langle\omega, x'\rangle + \gamma)q(\omega)\,\mathrm{d}\omega \,\mathrm{d}\gamma\label{eq:K_FT}\\
&\coloneqq\int_{\Phi=\mathcal{X}\times[0,2\pi)}\psi(x,\nu)\psi(x',\nu)\,\mathrm{d}\pi(\nu),
\nonumber
\end{align}
with $\psi(x,\nu) \coloneqq \sqrt{2}\cos(\langle\omega, x\rangle + \gamma)$, and $\pi = q\times\mu$ where $\mu$ is the uniform measure over $[0,2\pi)$. As such, we have indeed arrived at an integral representation for $K$. Given this derivation, we note that the name \textit{random Fourier features} comes from the fact that the measure $\pi$ is proportional to the Fourier transform of $\overline{K}$.

\subsection{PQC models for variational QML}\label{ss:PQC}

As discussed in Section~\ref{s:introduction}, variational QML is based on the classical optimization of models defined via \emph{parameterized quantum circuits} (PQCs)~\cite{Benedetti_2019}. In the context of regression, one begins by fixing a parameterized quantum circuit $C$, whose gates can depend on both data points $x\in\mathcal{X}$ and  components of a vector $\theta\in\Theta$ of variational parameters, where typically $\Theta= [0,2\pi)^c$ for some $c$. For each data point $x$ and each vector of variational parameters $\theta$, this circuit realizes the unitary $U(x,\theta)$. Given this, we then choose an observable $O$, and define the associated PQC model class $\mathcal{F}_{(C,O)}$ as the set of all functions $f_\theta\colon\mathcal{X}\rightarrow \mathbb{R}$ defined via
\begin{equation}
f_{\theta}(x) = \langle 0|U^{\dagger}(x,\theta)OU(x,\theta)|0\rangle
\end{equation}
for all $x\in\mathcal{X}$, i.e.,  
\begin{equation}
\mathcal{F}_{(C,O)} = \{f_\theta(\cdot) = \langle 0|U^{\dagger}(\cdot,\theta)OU(\cdot,\theta)|0\rangle\,|\, \theta\in \Theta\}.
\end{equation}
One then proceeds by using a classical optimization algorithm to optimize over the variational parameters $\theta$. In this work, we consider an important sub-class of PQC models in which the classical data $x$ enters only via Hamiltonian time evolutions, whose duration is controlled by a single component of $x$. To be more precise, for $x = (x_1,\ldots,x_d)\in\mathbb{R}^d$, we assume that each gate in the circuit $C$ which depends on $x$ is of the form
\begin{equation}
V_{(j,k)}(x_j) = e^{-iH^{(j)}_kx_j},
\end{equation}
for some Hamiltonian $H^{(j)}_k$. We stress that we exclude here more general encoding schemes, such as those which allow for time evolutions parameterized by functions of $x$, or time evolutions of parameterized linear combinations of Hamiltonians.   We denote by $\mathcal{D}^{(j)} = \{H^{(j)}_k\,|\, k\in [L_j]\}$ the set of all $L_j$ Hamiltonians which are used to encode the component $x_j$ at some point in the circuit, and we call the tuple
\begin{equation}
\mathcal{D} \coloneqq \left(\mathcal{D}^{(1)},\ldots,\mathcal{D}^{(d)}\right)
\end{equation}
the \textit{data-encoding strategy}. It is by now well known that these models admit a succinct ``classical'' description~\cite{data_encoding,gil2020input, Caro_2021} given by
\begin{align}\label{eq:fourier_rep}
f_\theta(x) = \sum_{\omega\in\tilde{\Omega}_\mathcal{D}} c_{\omega}(\theta) e^{i\langle\omega, x\rangle},
\end{align}
where 
\begin{enumerate}
\item the set of frequency vectors $\tilde{\Omega}_\mathcal{D}\subseteq\mathbb{R}^d$ is completely determined by the data-encoding strategy. We describe the construction of $\tilde{\Omega}_\mathcal{D}$ from $\mathcal{D}$ in Appendix~\ref{app:freq_construction}. 
\item the frequency coefficients $c_{\omega}(\theta)$ depend on the trainable parameters $\theta$, but in a way which usually does not admit a concise expression.
\end{enumerate}
As described in Ref.~\cite{data_encoding}, we know that $\omega_0\coloneqq(0,\ldots, 0)\in \tilde{\Omega}_\mathcal{D}$ and that the non-zero frequencies in $\tilde{\Omega}_\mathcal{D}$ come in mirror pairs -- i.e., $\omega\in\tilde{\Omega}_\mathcal{D}$ implies $-\omega\in\tilde{\Omega}_\mathcal{D}$. Additionally, one has $c_\omega(\theta) = c^*_{-\omega}(\theta)$ for all $\omega\in\tilde{\Omega}_\mathcal{D}$ and all $\theta$, which ensures the function $f_\theta$ evaluates to a real number. As a result, we can perform an arbitrary splitting of pairs to redefine $\tilde{\Omega}_\mathcal{D} \coloneqq\Omega_{\mathcal{D}}\cup \left(-\Omega_{\mathcal{D}}\right)$, where $\Omega_{\mathcal{D}}\cap\left(-\Omega_{\mathcal{D}}\right) = \{\omega_0\}$. It will also be convenient to define $\Omega^{+}_\mathcal{D}\coloneqq \Omega_\mathcal{D}\setminus\{\omega_0\}$. Given this, by defining
\begin{align}
a_\omega(\theta) &\coloneqq c_\omega(\theta) + c_{-\omega}(\theta), \\
b_\omega(\theta) &\coloneqq i(c_\omega(\theta) - c_{-\omega}(\theta))
\end{align}
for all $\omega\in\Omega^{+}_{\mathcal{D}}$, and writing $\Omega_\mathcal{D} = \{\omega_0,\omega_1,\ldots,\omega_{|\Omega^{+}_\mathcal{D}|}\},$ we can rewrite Eq.~\eqref{eq:fourier_rep} as
\begin{align}
f_\theta(x) &= c_{\omega_0}(\theta) + \sum_{i = 1}^{|\Omega^{+}_\mathcal{D}|}\left(a_{\omega_i}(\theta)\cos(\langle \omega_i,x\rangle) + b_{\omega_i}(\theta)\sin(\langle \omega_i,x\rangle)\right)
\nonumber
\\
&= \langle c(\theta) , \phi_\mathcal{D}(x)\rangle\label{eq:linear_rep}
\end{align}
where
\begin{align}
c(\theta) &\coloneqq \sqrt{|\Omega_\mathcal{D}|}\left(c_{\omega_0}(\theta),a_{\omega_1}(\theta),b_{\omega_1}(\theta),\ldots,a_{\omega_{|\Omega^{+}_\mathcal{D}|}}(\theta),b_{\omega_{|\Omega^{+}_\mathcal{D}|}}(\theta) \right),\\
\phi_\mathcal{D}(x) &\coloneqq \frac{1}{\sqrt{|\Omega_\mathcal{D}|}}\bigg(1,\cos(\langle \omega_1, x\rangle), \sin(\langle \omega_1, x\rangle),\ldots,\cos(\langle \omega_{|\Omega^{+}_\mathcal{D}|}, x\rangle),\sin(\langle \omega_{|\Omega^{+}_\mathcal{D}|}, x\rangle)\bigg),\label{eq:PQC_feature_map}
\end{align}
and the normalization constant has been chosen to ensure that $\langle \phi_\mathcal{D}(x),\phi_\mathcal{D}(x)\rangle = 1$, which will be required shortly. The formulation in Eq.~\eqref{eq:linear_rep} makes it clear that $f_{\theta}$ is a \textit{linear} model in $\mathbb{R}^{|\tilde{\Omega}_{\mathcal{D}}|}$, taken with respect to the feature map ${\phi_\mathcal{D}\colon\mathcal{X}\rightarrow\mathbb{R}^{|\tilde{\Omega}_{\mathcal{D}}|}}$. We can now define the model class of all linear models realizable by the parameterized quantum circuit with data-encoding strategy $\mathcal{D}$, variational parameter set $\Theta$ and observable $O$ via
\begin{align}
\mathcal{F}_{(\Theta,\mathcal{D},O)}  &= \{ f_\theta(\cdot) = \langle 0|U^{\dagger}(\theta,\cdot)OU(\theta,\cdot)|0\rangle\,|\, \theta \in \Theta\} \\
&= \{f_{\theta}(\cdot) = \langle c(\theta) , \phi_\mathcal{D}(\cdot)\rangle \,|\, \theta \in \Theta\}.
\end{align}
In what follows, we will use a tuple $(\Theta,\mathcal{D},O)$ to represent a PQC architecture, as we have done above. We note that for all $f\in \mathcal{F}_{(\Theta,\mathcal{D},O)}$ we have that $\|f\|_\infty \leq \|O\|_\infty$.

It is critical to note that due to the constraints imposed by the circuit architecture $(\Theta,\mathcal{D},O)$, the class $\mathcal{F}_{(\Theta,\mathcal{D},O)}$ may not contain \textit{all} possible linear functions with respect to the feature map $\phi_\mathcal{D}$. Said another way, the circuit architecture gives rise to an \textit{inductive bias} in the set of functions which can be realized. However, for the analysis that follows, it will be very useful for us to define the set of \textit{all} linear functions $\mathcal{F}_\mathcal{D}$ with respect to the feature map $\phi_\mathcal{D}$, i.e.,
\begin{equation}\label{eq:closure}
\mathcal{F}_{\mathcal{D}} = \left\{f_v(\cdot) = \langle v, \phi_\mathcal{D}(\cdot)\rangle \,|\, v\in\mathbb{R}^{|\tilde{\Omega}_{\mathcal{D}}|}\right\}.
\end{equation}
As shown in Figure.~\ref{fig:sets}, we stress that for any architecture $(\Theta,\mathcal{D},O)$, we have that
\begin{equation}\label{eq:first_inclusion}
\mathcal{F}_{(\Theta,\mathcal{D},O)} \subset \mathcal{F}_{\mathcal{D}}. 
\end{equation}
The inclusion is strict due to the fact that for all $f\in \mathcal{F}_{(\Theta,\mathcal{D},O)}$ we know that $\|f\|_\infty \leq \|O\|_\infty$, whereas % where as
$\mathcal{F}_{\mathcal{D}}$ contains functions of arbitrary infinity norm. However, we note that if one defines the set
\begin{equation}
\mathcal{F}_{(\mathcal{D},O)}= \left\{f\in \mathcal{F}_\mathcal{D} \,|\, \|f\|_{\infty}\leq \|O\|_{\infty}\right\},
\end{equation}
then in principle there could exist architectures for which $\mathcal{F}_{(\Theta,\mathcal{D},O)} =\mathcal{F}_{(\mathcal{D},O)}$, and therefore one has $\mathcal{F}_{(\Theta,\mathcal{D},O)} \subseteq\mathcal{F}_{(\mathcal{D},O)}$ for all architectures $(\Theta,\mathcal{D},O)$. As such, one can in some sense think of $\mathcal{F}_{(\mathcal{D},O)}$ as the ``closure'' of $\mathcal{F}_{(\Theta,\mathcal{D},O)}$.

\subsection{PQC feature map and PQC-kernel}\label{ss:PQC_kernels}

Given the observation from Section~\ref{ss:PQC} that all PQC models are linear in some high-dimensional feature space fully defined by the data-encoding strategy, we can very naturally associate to each data-encoding strategy both a feature map and an associated kernel function, which we call the \textit{PQC-kernel}:

\begin{definition}[PQC feature map and PQC-kernel]\label{def:PQC_kernel} Given a data-encoding strategy $\mathcal{D}$, we define the PQC feature map $\phi_\mathcal{D}\colon\mathcal{X}\rightarrow \mathbb{R}^{|\tilde{\Omega}_\mathcal{D}|}$ via Eq.~\eqref{eq:PQC_feature_map}, i.e.,
\begin{equation}
\phi_\mathcal{D}(x) \coloneqq \frac{1}{\sqrt{|\Omega_\mathcal{D}|}}\left(1,\cos(\langle \omega_1, x\rangle), \sin(\langle \omega_1, x\rangle),\ldots,\cos(\omega_{|\Omega^+_{\mathcal{D}}|}, x\rangle),\sin(\omega_{|\Omega^+_{\mathcal{D}}|}, x\rangle)\right)
\end{equation}
for $\omega_i\in \Omega^{+}_\mathcal{D}$. We then define the PQC-kernel $K_\mathcal{D}$ via
\begin{equation}
K_{\mathcal{D}}(x,x') \coloneqq \langle \phi_\mathcal{D}(x),\phi_\mathcal{D}(x')\rangle.
\end{equation}
\end{definition}
It is crucial to stress that the classical PQC-kernel defined in Definition~\ref{def:PQC_kernel} is fundamentally \textit{different} from the so called ``quantum kernels'' often considered in QML -- see, for example, Refs.~\cite{mariaQMLkernels,beyondkernel} --  which are defined from a data-parameterized unitary $U(x)$ via $K(x,x') = \mathrm{Tr}[\rho(x)\rho(x')]$ with $\rho(x) = U(x)|0\rangle\langle 0|U^\dagger(x)$. Additionally, we note that the feature map $\phi_\mathcal{D}$ defined in Definition~\ref{def:PQC_kernel} is \textit{not} the unique feature map with the property that all functions in $\mathcal{F}_{(\Theta,\mathcal{D},O)}$ are linear with respect to the feature map. Indeed, we will see in Section~\ref{ss:re-weight} that any ``re-weighting'' of $\phi_\mathcal{D}$ will preserve this property.

\section{Potential of RFF-based linear regression for dequantizing variational QML}\label{s:PQCvsRFF}

Let us now imagine that we have a regression problem to solve. More precisely, imagine that we have a dataset $S$, with $n$ elements drawn from some distribution $P$, as per Section~\ref{ss:stat_framework}. One option is for us to use hybrid quantum classical optimization. More specifically, we choose a PQC circuit architecture $(\Theta,\mathcal{D},O)$ -- consisting of data-encoding gates, trainable gates and measurement operator -- and then variationally optimize over the trainable parameters. We can summarize this as follows.

\begin{simplealgorithm}[Variational QML]\label{alg:VQML}
Choose a PQC architecture $(\Theta, \mathcal{D},O)$ and optimize over the parameters $\theta\in\Theta$. The output is some linear function $f_\theta \in \mathcal{F}_{(\Theta,\mathcal{D},O)}$.
\end{simplealgorithm}

We note that Algorithm~\ref{alg:VQML} essentially performs a variational search (typically via gradient based optimization) through $\mathcal{F}_{(\Theta,\mathcal{D},O)}$, which as per Lemma~\ref{lem:function_inclusions}, is some parameterized subset of $\mathcal{F}_\mathcal{D}$, the set of all linear functions with respect to the feature map $\phi_\mathcal{D}$. But this begs the question: Why run Algorithm~\ref{alg:VQML}, when we could just do classical linear regression with respect to the feature map $\phi_\mathcal{D}$? More specifically, instead of running Algorithm~\ref{alg:VQML}, why not just run the following purely \textit{classical} algorithm:

\begin{simplealgorithm}[Classical linear regression over $\mathcal{F}_{\mathcal{D}})$]\label{alg:LR} Given a PQC architecture $(\Theta, \mathcal{D},O)$, construct $\mathcal{D}$ and the feature map $\phi_\mathcal{D}$, and then perform linear regression with respect to the feature map $\phi_\mathcal{D}$. The output is some $f_v\in \mathcal{F}_{\mathcal{D}}$.
\end{simplealgorithm}
Unfortunately, Algorithm~\ref{alg:LR} has the following shortcomings:
\begin{enumerate}
\item \textbf{Exponential complexity}: Recall that $\phi_\mathcal{D}\colon\mathcal{X}\rightarrow \mathbb{R}^{|\tilde{\Omega}_{\mathcal{D}}|}$. As such, the space and time complexity of Algorithm~\ref{alg:LR} is $\mathcal{O}(n|\tilde{\Omega}_\mathcal{D}|)$ and $\mathcal{O}(n|\tilde{\Omega}_\mathcal{D}|^2 + |\tilde{\Omega}_\mathcal{D}|^3)$, respectively. Unfortunately, as detailed in Table 1  of Ref.~\cite{Caro_2021} and discussed in Section~\ref{ss:implementation}, the Cartesian product structure of $\tilde{\Omega}_\mathcal{D}$ results in a curse of dimensionality which leads to $|\tilde{\Omega}_\mathcal{D}|$ scaling \textit{exponentially} with respect to the dimension of the input data $d$ (which gives the size of the problem). For example, if one uses a data encoding strategy consisting only of Pauli Hamiltonians, and if each component is encoded via $L$ encoding gates, then one obtains $|\tilde{\Omega}_\mathcal{D}| = \mathcal{O}\left(L^d\right)$.
\item \textbf{Potentially poor generalization}: As we have noted in Eq.~\eqref{eq:first_inclusion}, and illustrated in Figure.~\ref{fig:sets}, due to the constrained depth/expressivity of the trainable parts of any PQC architecture which uses the data-encoding strategy $\mathcal{D}$, we have that
\begin{equation}\label{eq:inclusion}
\mathcal{F}_{(\Theta,\mathcal{D},O)} \subset\mathcal{F}_{\mathcal{D}},
\end{equation}
i.e., that $\mathcal{F}_{(\Theta,\mathcal{D},O)}$ is a \textit{subset} of $\mathcal{F}_{\mathcal{D}}$. Now, let the output of Algorithm~\ref{alg:VQML} be some $f_\theta \in \mathcal{F}_{(\Theta,\mathcal{D},O)}$ and the output of Algorithm~\ref{alg:LR} be some $f_v\in \mathcal{F}_{\mathcal{D}}$. Due to the fact that linear regression is a perfect empirical risk minimizer, and the inclusion in Eq.~\eqref{eq:inclusion}, we are guaranteed that
\begin{equation}
\hat{R}(f_v)\leq \hat{R}(f_\theta),
\end{equation}
i.e., the \textit{empirical risk} achieved by Algorithm~\ref{alg:LR} will always be better than the empirical risk achieved by Algorithm~\ref{alg:VQML}. However, it could be that Algorithm~\ref{alg:LR} ``overfits'' -- more precisely, it could be the case that the \textit{true risk} achieved by $f_\theta$ is better than that achieved by $f_v$, i.e., that
\begin{equation}\label{eq:true_risk_relative}
R(f_v)\geq R(f_\theta).
\end{equation}
Said another way, the PQC architecture results in an \textit{inductive bias}, which constrains the set of linear functions which are accessible to Algorithm~\ref{alg:VQML}. As illustrated in Figure~\ref{fig:sets}, it could be the case that this inductive bias leads the output of Algorithm~\ref{alg:VQML} to generalize better than that of Algorithm~\ref{alg:LR}. 
\end{enumerate}

In light of the above, the natural question is then as follows.

\begin{figure}[t]
\begin{center}
\includegraphics[width=\linewidth]{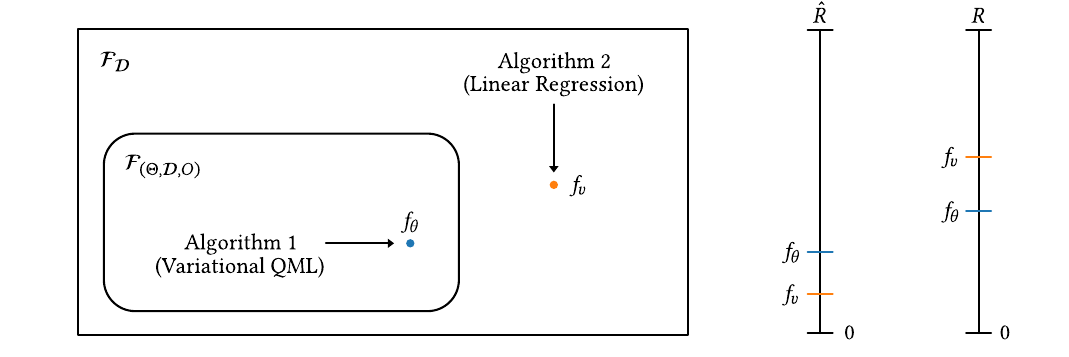} 
\caption{Illustration of the relationship between $\mathcal{F}_\mathcal{D}$ and $\mathcal{F}_{(\Theta,\mathcal{D},O)}$, and the output of Algorithms~\ref{alg:VQML} and~\ref{alg:LR}. In particular, we always have that $\mathcal{F}_{(\Theta,\mathcal{D},O)}\subset \mathcal{F}_\mathcal{D}$. As a consequence of this, and the fact that Algorithm~\ref{alg:LR} is a perfect empirical risk minimizer, we always have that $\hat{R}(f_v)\leq \hat{R}(f_\theta)$. However, it might be the case that $R(f_v)\geq R(f_\theta)$.}
\label{fig:sets}
\end{center}
\end{figure}

\begin{question}[Existence of efficient linear regression]
Can we modify Algorithm~\ref{alg:LR} (classical linear regression) such that it is both efficient with respect to $d$, and with high probability outputs a function which is just as good as the output of Algorithm~\ref{alg:VQML} (variational QML), with respect to the true risk?
\end{question}

As we have already hinted at in the preliminaries, one possible approach -- at least for addressing the issue of poor complexity -- is to use random Fourier features to approximate the classical PQC-kernel $K_\mathcal{D}(x,x') = \langle\phi_\mathcal{D}(x), \phi_\mathcal{D}(x')\rangle$ which is used implicitly in Algorithm~\ref{alg:LR}. Indeed, this has already been suggested and explored in Ref.~\cite{RFF}. More specifically Ref.~\cite{RFF}, at a very high level, suggests the following algorithm:

\begin{simplealgorithm}[Classical linear regression over $\mathcal{F}_{\mathcal{D}}$ with random Fourier features]\label{alg:RFF}Given a data-encoding strategy $\mathcal{D}$, implement RFF-based $\lambda$-regularized linear regression using the classical PQC-kernel $K_\mathcal{D}$, and obtain some function $h_\nu\in\mathcal{F}_\mathcal{D}$. More specifically:
\begin{enumerate}
\item Sample $M$ ``features'' $\nu_i = (\omega_i,\gamma_i)\in \mathbb{R}^d\times[0,2\pi)$ from the distribution $\pi = p_\mathcal{D}\times\mu$, which as per Eq.~\eqref{eq:K_FT} is the product distribution appearing in the integral representation of the shift-invariant kernel $K_\mathcal{D}$.
\item Construct the randomized feature map $\tilde{\phi}_M(x) = \frac{1}{\sqrt{M}}\left(\psi(x,\nu_1),\ldots,\psi(x,\nu_M)\right)$, where $\psi(x,\nu) = \sqrt{2}\cos(\langle\omega, x\rangle + \gamma)$. 
\item Implement $\lambda$-regularized linear regression with respect to the feature map $\tilde{\phi}_M$.
\end{enumerate}
\end{simplealgorithm}

In the above description of the algorithm we have omitted details of how to sample frequencies from the distribution $p_\mathcal{D}$, which will depend on the kernel $K_\mathcal{D}$. We note that in order for Algorithm~\ref{alg:RFF} to be efficient with respect to $d$, it is necessary for this sampling procedure to be efficient with respect to $d$. For simplicity, we assume for now that this is the case, and postpone a detailed discussion of this issue to Section~\ref{ss:implementation}. As discussed in Section~\ref{ss:RFF}, the space and time complexity of Algorithm~\ref{alg:RFF} (post-sampling) is $\mathcal{O}(nM)$ and $\mathcal{O}(nM^2 + M^3)$,  respectively, where $M$ is the number of frequencies which have been sampled. Given this setup, the natural question is then as follows:

\begin{question}[Efficiency of RFF for PQC dequantization]\label{q:efficiency} Given a regression problem $P$ and a circuit architecture $(\Theta,\mathcal{D},O)$, how many data samples $n$ and frequency samples $M$ are necessary to ensure that, with high probability, the true risk of the function $h_v$ output by Algorithm~\ref{alg:RFF} (RFF) is no more than $\epsilon$ worse than the true risk of the function $f_\theta$ output by Algorithm~\ref{alg:VQML} (variational QML) -- i.e., to ensure that $R(h_\nu) - R(f_\theta)\leq \epsilon$?
\end{question}

Said another way, Question~\ref{q:efficiency} is asking when classical RFF-based regression (Algorithm~\ref{alg:RFF}) can be used to efficiently \textit{dequantize} variational QML (Algorithm~\ref{alg:VQML}). In Ref.~\cite{RFF} the authors addressed a similar question, but with two important differences:

\begin{enumerate}
\item It was implicitly assumed that $(\Theta,\mathcal{D},O)$ is universal -- i.e.,  that using the PQC one can realize all functions in $\mathcal{F}_{(\mathcal{D},O)}$. Recall however from the discussion above that we are precisely interested in the case in which $(\Theta,\mathcal{D},O)$ is not universal -- i.e.,  the case in which $\mathcal{F}_{(\Theta,\mathcal{D},O)} \subset \mathcal{F}_{(\mathcal{D},O)}$ due to constraints on the circuit architecture. This is for two reasons: Firstly, because this is the case for practically realizable near-term circuit architectures. Secondly, because it is in this regime in which the circuit architecture induces an \textit{inductive bias} which may lead to better generalization than the output of linear regression over $\mathcal{F}_\mathcal{D}$.
\item Instead of considering true risk, Ref.~\cite{RFF} considered the complexity necessary to achieve ${|h_\nu(x) - f_\theta(x)| \leq \epsilon}$ for all $x$. We note that this is \textit{stronger} than $|R(h_\nu) - R(f_\theta)|\leq \epsilon$, due to the latter comparing the functions with respect to the data distribution $P$. Here we are concerned with the latter, which is the typical goal in statistical learning theory.
\end{enumerate}

In light of these considerations, we proceed in Section~\ref{ss:RFF_analysis} to provide answers to Question~\ref{q:efficiency}.

\textbf{Note on recent related work:} As per the discussion above, the two motivations for introducing Algorithm~\ref{alg:RFF} were the poor efficiency and potentially poor generalization associated to Algorithm~\ref{alg:LR}. However, we note that in Ref.~\cite{shin2023analyzing}, which appeared recently, the authors show via a tensor network (TN) analysis that to every PQC architecture one can associate a feature map -- different from the PQC feature map $\phi_\mathcal{D}$ --  for which (a) all functions in the PQC model class are linear with respect to the feature map, and (b) the associated kernel \textit{can} be evaluated efficiently classically. As such, using this feature map, for any number of data-samples $n$, one \textit{can} run Algorithm~\ref{alg:LR} classically efficiently with respect to $d$ -- i.e.,  there is no need to approximate the kernel via RFF! Indeed, this approach to dequantization is suggested by the authors of Ref.~\cite{shin2023analyzing}. However as above,  and discussed in Ref.~\cite{shin2023analyzing}, this does not immediately yield an efficient dequantization of variational QML, due to the potentially poor generalization of linear regression over the entire function space. In light of this, the generalization of linear regression with respect to the tensor network kernel introduced in Ref.~\cite{shin2023analyzing} certainly deserves attention, and we hope that the methods and tools introduced in this work, and others such as Ref.~\cite{khavari2021lower}, can be useful in that regard.

\section{Generalization and efficiency guarantees for RFF-based linear regression}\label{ss:RFF_analysis}

In this section, we attempt to provide rigorous answers to Question~\ref{q:efficiency} -- i.e.,  for which PQC architectures and for which regression problems does RFF-based linear regression yield an efficient dequantization of variational QML?

In particular, this section is structured as follows: We begin in Section~\ref{ss:digression} with a brief digression, providing definitions for some important kernel notions. With these in hand, we continue in Section~\ref{ss:RFF_efficiency} to state Theorem~\ref{thm:efficiency_RFF} which provides a concrete answer to Question~\ref{q:efficiency}. In particular, Theorem~\ref{thm:efficiency_RFF} provides an upper bound on both the number of data samples $n$, and the number of frequency samples $M$, which are sufficient to ensure that, with high probability, the output of RFF-based linear regression with respect to the kernel $K_\mathcal{D}$ is a good approximation to the best possible PQC function, with respect to true risk. As we will see, these upper bounds depend crucially on two quantities which are defined in Section~\ref{ss:digression}, namely the operator norm of the kernel integral operator associated to $K_\mathcal{D}$, and the reproducing kernel Hilbert space norm of the optimal PQC function. 

Given the results of Section~\ref{ss:RFF_efficiency}, in order to make any concrete statements it is necessary to gain a deeper quantitative understanding  of both the operator norm of the kernel integral operator and the RKHS norm of functions in the PQC model class. However, before doing this, we show in Section~\ref{ss:re-weight} that the PQC feature map $\phi_\mathcal{D}$ is in fact a special instance of an entire family of feature maps -- which we call ``re-weighted PQC feature maps'' -- and that Theorem~\ref{thm:efficiency_RFF} holds not only for $K_\mathcal{D}$, but for the kernel induced by any such re-weighted feature map. With this in hand, we then proceed in Section~\ref{ss:implementation} to show how the re-weighting determines the distribution over frequencies $\pi$ from which one needs to sample in the RFF procedure, and we discuss in detail for which feature maps/kernels one can and cannot \textit{efficiently} sample from this distribution. This immediately allows us to rule out efficient RFF dequantization for a large class of re-weighted PQC feature maps, and therefore allows us to focus our attention on only those feature maps (re-weightings) which admit efficient sampling procedures.

With this knowledge, we then proceed in Sections~\ref{ss:kern_int_op} and~\ref{ss:RKHS_norm} to discuss in detail the quantitative behaviour of the kernel integral operator and the RKHS norm of PQC functions, for different re-weighted PQC kernels. This allows us to place further restrictions on the circuit architectures and re-weightings which yield efficient PQC dequantization via Theorem~\ref{thm:efficiency_RFF}. Finally, in Section~\ref{ss:lower_bounds} we show that the properties we have identified in Sections~\ref{ss:kern_int_op} and~\ref{ss:RKHS_norm} as \textit{sufficient} for the application of Theorem~\ref{thm:efficiency_RFF} are in some sense \textit{necessary}. In particular, we prove \textit{lower bounds} on the number of frequency samples necessary to achieve a certain \textit{average} error via RFF, and use this to delineate rigorously when efficient dequantization via RFF is \textit{not} possible.

\subsection{Preliminary kernel theory}\label{ss:digression}

In order to present Theorem~\ref{thm:efficiency_RFF} we require a few definitions. To start, we need the notion of a \emph{reproducing kernel Hilbert space} (RKHS), and the associated RKHS norm.

\begin{definition}[RKHS and RKHS norm] Given a kernel $K\colon\mathcal{X}\times\mathcal{X}\rightarrow \mathbb{R}$ we define the associated reproducing kernel Hilbert space (RKHS) as the tuple $(\mathcal{H}_K,\langle\cdot,\cdot\rangle_K)$ where $\mathcal{H}_K$ is the set of functions defined as the completion (including limits of Cauchy series) of
\begin{equation}
\mathrm{span}\{K_x(\cdot) \coloneqq K(x,\cdot)\,|\, x\in \mathcal{X}\},
\end{equation}
and $\langle\cdot,\cdot\rangle_K$ is the inner product on $\mathcal{H}_K$ defined via
\begin{equation}\label{eq:RKHS_inner}
\langle g, h\rangle_{\mathcal{H}_K} \coloneqq \sum_{i,j}\alpha_i \beta_j K(x_i,x_j)
\end{equation}
for any two functions $g = \sum_{i}\alpha_iK_{x_i}$ and $h = \sum_{j}\beta_jK_{x_j}$ in $\mathcal{H}_{K}$. This inner product then induces the RKHS norm $\|\cdot\|_K$ defined via
\begin{equation}
\|g\|_{K} \coloneqq \sqrt{\langle g, g\rangle_{\mathcal{H}_K}}.
\end{equation}
\end{definition}

It is crucial to note that for two kernels $K_1$ and $K_2$ it may be the case that $\mathcal{H}_{K_1} = \mathcal{H}_{K_2}$ but $\|\cdot\|_{K_1}\neq \|\cdot\|_{K_2}$. We will make heavy use of this fact shortly. In particular, the re-weighted PQC kernels introduced in Section~\ref{ss:re-weight} have precisely this property. In addition to the definition above, we also need a definition of the kernel integral operator associated to a kernel, which we define below.

\begin{definition}[Kernel integral operator]\label{def:kern_int_op} Given a kernel $K\colon\mathcal{X}\times\mathcal{X}\rightarrow \mathbb{R}$ and a probability distribution $P_\mathcal{X}$ over $\mathcal{X}$, we start by defining the space of square-integrable functions with respect to $P_\mathcal{X}$ via
\begin{equation}
L^2(\mathcal{X},P_\mathcal{X}
) = \left\{f\in \mathbb{R}^{\mathcal{X}}\text{ such that } \int_\mathcal{X}|f(x)|^2\,\mathrm{d}P_\mathcal{X}(x)<\infty\right\}.
\end{equation}
The kernel integral operator $T_K\colon L^2(\mathcal{X},P_\mathcal{X}
)\rightarrow L^2(\mathcal{X},P_\mathcal{X}
)$ is then defined via
\begin{equation}
(T_Kg)(x) = \int_\mathcal{X}K(x,x')g(x')\,\mathrm{d}P_\mathcal{X}(x')
\end{equation}
for all $g\in L^2(\mathcal{X},P_\mathcal{X})$.
\end{definition}
In addition, we note that we will mostly be concerned with the \textit{operator norm} of the kernel integral operator, which we denote with $\|T_{K}\|$ -- i.e., when no subscript is used to specify the norm, we assume the operator norm.

\subsection{Efficiency of RFF for matching variational QML}\label{ss:RFF_efficiency}

Given the definitions of Section~\ref{ss:digression}, we proceed in this section to state Theorem~\ref{thm:efficiency_RFF}, which provides insight into the number of data samples $n$ and the number of frequency samples $M$ which are sufficient to ensure that, with probability at least $1-\delta$, the true risk of the output hypothesis of RFF-based linear regression is no more than $\epsilon$ worse than the true risk of the best possible function realizable by the PQC architecture. To this end, we require first a preliminary definition of the ``best possible PQC function'':

\begin{definition}[Optimal PQC function] Given a regression problem $P\sim\mathcal{X}\times \mathbb{R}$, and a PQC architecture $(\Theta,\mathcal{D},O)$, we define $f^*_{(\Theta,\mathcal{D},O)}$, the optimal PQC model for $P$ (with respect to \textit{true risk}), via
\begin{equation}
f^*_{(\Theta,\mathcal{D},O)} = \argmin_{f\in \mathcal{F}_{(\Theta,\mathcal{D},O)}}\left[R(f)\right].
\end{equation}
\end{definition}

With this in hand we can state Theorem~\ref{thm:efficiency_RFF}. We note however that this follows via a straightforward application of the RFF generalization bounds provided in Ref.~\cite{genRFF} to the PQC-kernel, combined with the insight that the set of PQC functions $\mathcal{F}_{(\Theta,\mathcal{D},O)}$ is contained within $\mathcal{H}_{K_\mathcal{D}}$, the RKHS associated with the PQC kernel. A detailed proof is provided in Appendix~\ref{app:proof_main}. Additionally we note that the generalization bound given in Theorem~\ref{thm:efficiency_RFF} below holds for 2-norm \textit{regularized} RFF-based linear regression -- as per Eq.~\eqref{eq:lrr_reg_risk} -- with a a very specific regularization $\lambda = 1/\sqrt{n}$. This regularization is inherited from the proof of the generalization bounds in Ref.~\cite{genRFF}, and we refer to there for details.

\begin{theorem}[RFF vs.~variational QML]\label{thm:efficiency_RFF} Let $R$ be the risk associated with a regression problem $P\sim\mathcal{X}\times\mathbb{R}$. Assume the following:
\begin{enumerate}
\item $\|f^*_{(\Theta,\mathcal{D},O)}\|_{K_\mathcal{D}} \leq C$
\item $|y|\leq b$ almost surely when $(x,y)\sim P$, for some $b>0$.
\end{enumerate}
Additionally, define
\begin{align}
n_0 &\coloneqq \max\left\{4\|T_{K_\mathcal{D}}\|^2, \left(528\log\frac{1112\sqrt{2}}{\delta}\right)^2\right\},\\
c_0&\coloneqq 36\left(3 + \frac{2}{\|T_{K_\mathcal{D}}\|}\right), \\
c_1&\coloneqq 8\sqrt{2}(4b + \frac{5}{\sqrt{2}}C + 2\sqrt{2C}).
\end{align}
Then, let $\delta\in (0,1]$, $\epsilon> 0$, $n\geq n_0$, set $\lambda_n = 1/\sqrt{n}$, and let $\hat{f}_{M_n,\lambda_n}$ be the output of $\lambda_n$-regularized linear regression with respect to the feature map
\begin{equation}\label{eq:new_approx_map}
\phi_{M_n}(x) = \frac{1}{\sqrt{M_n}}\big(\psi(x,\nu_1),\ldots,\psi(x,\nu_{M_n})\big)
\end{equation}
constructed from the integral representation of $K_\mathcal{D}$ by sampling $M_n$ elements from $\pi$. Then, with probability at least $1-\delta$, one achieves 
\begin{equation}
R(\hat{f}_{M_n,\lambda_n}) - R(f^*_{(\Theta,\mathcal{D},O)}) \leq \epsilon,
\end{equation}
by ensuring that 
\begin{equation}
n\geq \max \left\{\frac{c_1^2\log^4\frac{1}{\delta}}{\epsilon^2}, n_0 \right\}
\end{equation}
and 
\begin{equation}
M \geq c_0\sqrt{n}\log\frac{108\sqrt{n}}{\delta}. 
\end{equation}
\end{theorem}

Let us now try to unpack what insights can be gained from Theorem~\ref{thm:efficiency_RFF}. Firstly, recall that here we consider $\mathcal{X} \subseteq \mathbb{R}^{d}$, in which case $d$ provides the relevant asymptotic scaling parameter. With this in mind, in order to gain intuition, assume for now that $n_0,c_0$ and $c_1$ are constants with respect to $d$. Assume additionally that one can sample efficiently from $\pi$. In this case we see via Theorem~\ref{thm:efficiency_RFF} that both the number of data points $n$, and the number of samples $M$, which are sufficient to ensure that with probability $1-\delta$ there is a gap of at most $\epsilon$ between the true risk of PQC optimization and the true risk of RFF, is independent of $d$, and polylogarithmic in both $1/\epsilon$ and $1/\delta$. Given that the time and space complexity of RFF-based linear regression is $\mathcal{O}(nM)$ and $\mathcal{O}(nM^2 + M^3)$, respectively, in this case Theorem~\ref{thm:efficiency_RFF} guarantees that Algorithm~\ref{alg:RFF} (RFF-based linear regression) provides an efficient dequantization of variational QML.

However, in general $n_0,c_0$ and $c_1$ will \textit{not} be constants, and one will \textit{not} be able to sample efficiently from $\pi$. In particular, $n_0,c_0$ and $c_1$ depend on both $\|T_{K_\mathcal{D}}\|$, the operator norm of the kernel integral operator, and $C$, the upper bound on the RKHS norm of the optimal PQC function. Given the form of $n_0,c_0$ and $c_1$, in order for Theorem~\ref{thm:efficiency_RFF} to yield a polynomial upper bound on $n$ and $M$, it is sufficient that $\|T_{K_\mathcal{D}}\| = \Omega(1/\mathrm{poly}(d))$ and that $C = \mathcal{O}(\mathrm{poly}(d))$.

\textbf{Summary:} In order to use Theorem~\ref{thm:efficiency_RFF} to make any concrete statement concerning the efficiency of RFF for dequantizing variational QML, we need to obtain the following.
\begin{enumerate}
\item \textit{Lower bounds} on $\|T_{K_\mathcal{D}}\|$, the operator norm of the kernel integral operator.
\item \textit{Upper bounds} on  $\|f^*_{(\Theta,\mathcal{D},O)}\|_{K_\mathcal{D}}$, the RKHS norm of the optimal PQC function.
\item An understanding of the complexity of sampling from $\pi$, the distribution appearing in the integral representation of $K_\mathcal{D}$.
\end{enumerate}

We address the above requirements in the following sections. However, before doing that, we note in Section~\ref{ss:re-weight} below that Theorem~\ref{thm:efficiency_RFF} applies not only to $K_\mathcal{D}$, but to an entire family of ``re-weighted'' kernels, of which $K_\mathcal{D}$ is a specific instance. We will see that this is important as different re-weightings will lead to substantially different sampling complexities, as well as lower and upper bounds on $\|T_{K_\mathcal{D}}\|$ and $C$, respectively.

\subsection{Re-weighted PQC kernels}\label{ss:re-weight}

As mentioned in Section~\ref{ss:RFF_efficiency}, the proof of Theorem~\ref{thm:efficiency_RFF} is essentially a straightforward application of the RFF generalization bound from Ref.~\cite{genRFF} to the classical PQC kernel $K_\mathcal{D}$. In particular, as shown in Appendix.~\ref{app:proof_main}, the key insight that allows one to leverage the generalization bound from Ref.~\cite{genRFF} into a bound on the difference in true risk between the output of variational QML and the output of RFF-based linear regression with the PQC kernel $K_\mathcal{D}$, is the fact that $\mathcal{F}_{(\Theta, \mathcal{D},O)}  \subseteq \mathcal{H}_{K_\mathcal{D}}$ -- i.e., the fact that the set of all PQC functions $\mathcal{F}_{(\Theta,\mathcal{D},O)}$ is contained within the RKHS associated with the PQC kernel $\mathcal{H}_{K_\mathcal{D}}$. With this in mind we can ask whether there exists any \textit{other} kernel $K$ for which $\mathcal{F}_{(\Theta, \mathcal{D},O)} \subseteq \mathcal{H}_{K}$, as Theorem~\ref{thm:efficiency_RFF} would then immediately also hold for RFF-based linear regression via $K$. The motivation for asking this question is the following.
\begin{enumerate}
\item As discussed in Section.~\ref{ss:RFF_efficiency} above, the number of sufficient data points $n$ and the number of sufficient samples $M$ specified by Theorem~\ref{thm:efficiency_RFF} depends on the RKHS norm and the kernel integral operator respectively, both of which depend on the kernel. As such, by using a different kernel, one may require less data-points $n$ and less samples $M$ to get the desired guarantee on the output of RFF-based linear regression.
\item The implementation of RFF based linear regression requires one to sample from the distribution $\pi$, which for shift invariant kernels is proportional to the Fourier transform of the kernel. In practice, sampling from this distribution for the kernel $K_\mathcal{D}$ may be hard, and using a different kernel might allow us to sample from a different distribution, for which efficient sampling is possible.
\end{enumerate}
In this section we show that there is indeed a whole family of kernels -- which we call \textit{re-weighted} PQC kernels -- for which the set of all PQC functions $\mathcal{F}_{(\Theta,\mathcal{D},O)}$ is contained within the RKHS of the kernel, and for which Theorem~\ref{thm:efficiency_RFF} is valid. This will turn out to be important and useful for a variety of reasons. Firstly, in Section~\ref{ss:implementation} we will show that indeed, sampling from the distribution $\pi$ associated with the original PQC kernel $K_{\mathcal{D}}$ can \textit{not} be done efficiently, whereas sampling from the distribution associated with certain re-weighted kernels can be done efficiently. Additionally, we then show in Sections~\ref{ss:kern_int_op} and \ref{ss:RKHS_norm} that for certain regression problems, using the original PQC kernel $K_\mathcal{D}$ for RFF-based linear will \textit{not} give an efficient dequantization method, whereas certain re-weighted kernels will.

We stress that in typical applications of RFF the kernel is fixed, however for the purposes of dequantization via RFF we do not necessarily need to restrict ourselves to any particular kernel and can therefore exploit the fact that Theorem~\ref{thm:efficiency_RFF} is valid for \textit{any} kernel $k$ for which  $\mathcal{F}_{(\Theta, \mathcal{D},O)}  \subseteq \mathcal{H}_{K}$ to \textit{choose} the optimal kernel for dequantization. Here we discuss the family of ``re-weighted" PQC kernels for which $\mathcal{F}_{(\Theta, \mathcal{D},O)}  \subseteq \mathcal{H}_{K}$, and in subsequent sections we provide insights into how to make a well-informed choice as to which kernel should be used for dequantization of a specific PQC architecture and regression problem.

Given this motivation, we can now make the notion of a re-weighted PQC kernel precise. For any ``re-weighting vector'' $\doubleyou\in\mathbb{R}^{|\Omega_\mathcal{D}|}$, we define the re-weighted PQC feature map via

\begin{align}
\phi_{(\mathcal{D},\doubleyou)}(x) &\coloneqq \frac{1}{\|\doubleyou\|_2}\big( \doubleyou_0,\doubleyou_1\cos(\langle \omega_1, x\rangle), \doubleyou_1\sin(\langle \omega_1, x\rangle),\ldots,\nonumber\\
&\qquad\qquad\qquad,\ldots,\doubleyou_{|\Omega^{+}_\mathcal{D}|}\cos(\langle \omega_{|\Omega^{+}_\mathcal{D}|}, x\rangle),\doubleyou_{|\Omega^{+}_\mathcal{D}|}\sin(\langle \omega_{|\Omega^{+}_\mathcal{D}|}, x\rangle)\big),
\end{align}
along with the associated set of linear functions with respect to $\phi_{(\mathcal{D},\doubleyou)}$, defined via
\begin{equation}
\mathcal{F}_{(\mathcal{D},\doubleyou)}  = \{f_{\theta}(\cdot) = \langle v , \phi_{(\mathcal{D},\doubleyou)}(\cdot)\rangle \,|\, v\in\mathbb{R}^{|\tilde{\Omega}_{\mathcal{D}}|}\},
\end{equation}
and the associated re-weighted PQC kernel
\begin{equation}\label{eq:re-weighted_kernel}
K_{(\mathcal{D},\doubleyou)}(x,x') \coloneqq \langle \phi_{(\mathcal{D},\doubleyou)}(x),\phi_{(\mathcal{D},\doubleyou)}(x')\rangle.
\end{equation}
Note that we recover the previous definitions when $\doubleyou$ is the vector of all $1$'s (in which case $\|\doubleyou\|_2 = \sqrt{|\Omega_\mathcal{D}|}$). With this in hand, as per the motivating discussion above, we would now like to show that there exists a set of re-weighting vectors $\doubleyou$ such that $\mathcal{F}_{(\Theta, \mathcal{D},O)}  \subseteq \mathcal{H}_{K_{(\mathcal{D},\doubleyou)}}$ for all $\doubleyou$ in the set. To this end, we start with the following observation: 

\begin{observation}[Invariance of $\mathcal{F}_{\mathcal{D}}$ under non-zero feature map re-weighting]\label{obs:reweighting} For a re-weighting vector $\doubleyou\in\mathbb{R}^{|\Omega_\mathcal{D}|}$ satisfying $\doubleyou_i\neq 0$ for all $i\in[|\Omega_\mathcal{D}|]$, we have that
\begin{equation}
\mathcal{F}_{(\mathcal{D},\doubleyou)} = \mathcal{F}_{\mathcal{D}}.
\end{equation}
\end{observation}
\begin{proof} Define the matrix $M_{\doubleyou}$ as the matrix with $\doubleyou$ as its diagonal, and the matrix
\begin{equation}
M\coloneqq \left(\sqrt{|\Omega|_\mathcal{D}}/\|\doubleyou\|_2\right)M_{\doubleyou}. 
\end{equation}
By the assumptions on $\doubleyou$, the matrix $M$ is invertible. Now, let $f \in \mathcal{F}_{(\mathcal{D},\doubleyou)}$. We have that 
\begin{equation}
f(\cdot) = \langle v,\phi_{(\mathcal{D},\doubleyou)}(\cdot)\rangle = \langle v,M\phi_{\mathcal{D}}(\cdot)\rangle = \langle vM,\phi_\mathcal{D}(\cdot)\rangle =  \langle \tilde{v},\phi_\mathcal{D}(\cdot)\rangle,
\end{equation}
i.e.,  $f\in \mathcal{F}_{\mathcal{D}}$, and, therefore,  $\mathcal{F}_{(\mathcal{D},\doubleyou)}\subset \mathcal{F}_{\mathcal{D}}$. Similarly, for all $g\in \mathcal{F}_{\mathcal{D}}$ we have that 
\begin{equation}
g(\cdot) = \langle v,\phi_\mathcal{D}(\cdot)\rangle = \langle v,M^{-1}M\phi_\mathcal{D}(\cdot)\rangle = \langle vM^{-1},M\phi_\mathcal{D}(\cdot)\rangle= \langle \tilde{v},\phi_{(\mathcal{D},\doubleyou)}(\cdot)\rangle,
\end{equation}
i.e.,  $g\in \mathcal{F}_{(\mathcal{D},\doubleyou)}$, and hence $\mathcal{F}_{\mathcal{D}} \subset\mathcal{F}_{(\mathcal{D},\doubleyou)}$.
\end{proof}
Now, the fact that for any feature map the set of all linear functions with respect to the feature map is a subset of the RKHS of the associated kernel~\cite{SVMbook}, tells us that $\mathcal{F}_{(\mathcal{D},\doubleyou)}  \subseteq \mathcal{H}_{K_{(\mathcal{D},\doubleyou)}}$ for all re-weighting vectors $\doubleyou$. Combing this with Observation~\ref{obs:reweighting} then gives
\begin{equation}
\mathcal{F}_{(\Theta, \mathcal{D},O)} \subset \mathcal{F}_\mathcal{D} = \mathcal{F}_{(\mathcal{D},\doubleyou)} \subseteq \mathcal{H}_{K_{(\mathcal{D},\doubleyou)}}.
\end{equation}
for all re-weighting vectors with no zero elements. As such we indeed have $\mathcal{F}_{(\Theta, \mathcal{D},O)}  \subseteq \mathcal{H}_{K_{(\mathcal{D},\doubleyou)}}$ for all $\doubleyou$ with no zero elements. As discussed above, we therefore also have that Theorem~\ref{thm:efficiency_RFF} actually holds for \textit{any} PQC kernel $K_{(\mathcal{D},\doubleyou)}$, re-weighted via a re-weighting vector $\doubleyou$ with no zero elements, and as such we can \textit{choose} which such kernel to use for our dequantization procedure. We note that allowing re-weighting vectors with zero elements has the effect of ``shrinking'' the set $F_{(\mathcal{D},\doubleyou)}$, which might result in the existence of functions $f$ satisfying both $f\in \mathcal{F}_{(\Theta, \mathcal{D},O)}$ and $f\notin F_{(\mathcal{D},\doubleyou)}$. Intuitively, this is problematic because for regression problems in which $f$ is the optimal solution, we know that the PQC architecture $(\Theta, \mathcal{D},O)$ can realize $f$, but we cannot hope for the RFF procedure via $K_{(\mathcal{D},\doubleyou)}$ to do the same, as it is limited to hypotheses within $F_{(\mathcal{D},\doubleyou)}$.

In light of the above observations, and as discussed above, from this point on we can broaden our discussion of the application of Theorem~\ref{thm:efficiency_RFF} to include all appropriately re-weighted PQC kernels. This insight is important because of the following.

\begin{enumerate}
\item We will see that while all appropriately re-weighted PQC kernels give rise to the same function set $\mathcal{F}_\mathcal{D}$, they give rise to \textit{different} RKHS norms. As a result, the RKHS norm of the optimal PQC function --  which as we have seen is critical to the complexity of the RFF method -- will depend on which re-weighting we choose.
\item Similarly, we will see that the operator norm of the kernel integral operator -- and therefore again the complexity of RFF linear regression -- depends heavily on the re-weighting chosen.
\item We will see that the re-weighting of the PQC kernel completely determines the probability distribution $\pi$, from which it is necessary to sample in order to implement RFF-based linear regression. Therefore, once again, the efficiency of RFF will depend on the re-weighting chosen. This perspective will also allow us to see why zero elements are not allowed in the re-weighting vector. Namely, because doing this will cause the probability of sampling the associated frequency to be zero, which is problematic if that frequency is required to represent the regression function.
\end{enumerate}

\subsection{RFF implementation}\label{ss:implementation}

Recall from Section~\ref{ss:RFF}, and from our presentation of Algorithm~\ref{alg:RFF} in Section~\ref{s:PQCvsRFF}, that when given a shift-invariant kernel $K$, in order to implement RFF-based linear regression, one has to sample from the probability measure ${\pi = p \,\times\, \mu}$, which one obtains from the Fourier transform of $K$. As such, given a re-weighted PQC kernel $K_{(\mathcal{D},\doubleyou)}$, we need to
\begin{enumerate}
\item understand the structure of the probability distribution $p$, which should depend on both $\mathcal{D}$ and the re-weighting vector $\doubleyou$, and
\item understand when and how -- i.e.,  for which data-encoding strategies and which re-weightings -- one can efficiently sample from $p$.
\end{enumerate}

Let us begin with point 1. To this end we start by noting that the re-weighted PQC kernels $K_{(\mathcal{D},\doubleyou)}$ have a particularly simple integral representation, from which one can read off the required distribution. In particular, note that
\begin{align}
K_{(\mathcal{D},\doubleyou)}(x,x') &= \langle \phi_{(\mathcal{D},\doubleyou)}(x),\phi_{(\mathcal{D},\doubleyou)}(x')\rangle 
\nonumber\\
&= \frac{1}{\|\doubleyou\|^2_2}\left(\doubleyou_0^2 + \sum_{i =1}^{|\Omega^{+}_\mathcal{D}|}\doubleyou_i^2\left[\cos(\langle\omega_i, x\rangle)\cos(\langle\omega_i, x'\rangle) + \sin(\langle\omega_i, x\rangle)\sin(\langle\omega_i, x'\rangle) \right]\right) \label{eq:almost_integral}\\
\nonumber
&= \frac{1}{\|\doubleyou\|^2_2}\sum_{i =0}^{|\Omega^{+}_\mathcal{D}|}\doubleyou_i^2\left[\cos(\langle\omega_i, x\rangle)\cos(\langle\omega_i, x'\rangle) + \sin(\langle\omega_i, x\rangle)\sin(\langle\omega_i, x'\rangle) \right]\\
\nonumber
&= \frac{1}{\|\doubleyou\|^2_2}\sum_{i =0}^{ |\Omega^{+}_\mathcal{D}|}\doubleyou_i^2\left[\cos(\langle\omega_i, (x-x')\rangle) \right],\\
& = \frac{1}{2\pi}\int_{\mathcal{X}}\int_0^{2\pi}\sqrt{2}\cos(\langle\omega, x\rangle + \gamma)\sqrt{2}\cos(\langle\omega, x'\rangle + \gamma)q_{(\mathcal{D},\doubleyou)}(\omega)\,\,\mathrm{d}\gamma \mathrm{d}\nu, \label{eq:final_integral}
\end{align}
where 
\begin{align}
q_{(\mathcal{D},\doubleyou)}(\omega) =\sum_{i = 0}^{|\Omega^+_\mathcal{D}|}\frac{\doubleyou_i^2}{\|\doubleyou\|^2_2}\delta(\omega - \omega_i),
\end{align}
and $\delta$ is the Dirac delta function. By comparison with Eq.~\eqref{eq:K_FT} we,  therefore,  see that $\pi=q_{(\mathcal{D},\doubleyou)}\times\mu$, where as before, $\mu$ is the uniform distribution over $[0,2\pi)$. For convenience, we refer to $q_{(\mathcal{D},\doubleyou)}$ as the probability distribution associated to $K_{(\mathcal{D},\doubleyou)}$.

Let us now move on to point 2 -- in particular, for which data encoding strategies and re-weighting vectors can we \textit{efficiently} sample from $q_{(\mathcal{D},\doubleyou)}$? Firstly, note that sampling from the \textit{continuous} distribution $q_{(\mathcal{D},\doubleyou)}$ can be done by sampling from the \textit{discrete} distribution $p_{(\mathcal{D},\doubleyou)}$ over $\Omega_{\mathcal{D}}$ defined via
\begin{align}
p_{(\mathcal{D},\doubleyou)}(\omega_i) =\frac{\doubleyou_i^2}{\|\doubleyou\|^2_2}\label{eq:freq_prob}
\end{align}
for all $\omega_i\in \Omega_{\mathcal{D}}$. As a result, from this point on we focus on the distribution $p_{(\mathcal{D},\doubleyou)}$, and when clear from the context  we drop the subscript and just use $p$ to refer to $p_{(\mathcal{D},\doubleyou)}$. Also, we note that we can in principle just work directly with the choice of probability distribution, as opposed to the underlying weight vector, as we know for any probability distribution over $\Omega_{\mathcal{D}}$ there exists an appropriate weight vector.

As $p$ is a distribution over $\Omega_{\mathcal{D}}$, in order to discuss the efficiency of sampling from $p$, it is necessary to briefly recall some facts about the sets $\Omega_\mathcal{D}$ and $\tilde{\Omega}_\mathcal{D}$. In particular, as discussed in Appendix~\ref{app:freq_construction}, given a data-encoding strategy $\mathcal{D}=\left(\mathcal{D}^{(1)},\ldots,\mathcal{D}^{(d)}\right)$, where $\mathcal{D}^{(j)}$ contains the Hamiltonians used to encode the $j$'th data component $x_j$, we know that 
\begin{equation}
\tilde{\Omega}_\mathcal{D} = \tilde{\Omega}^{(1)}_\mathcal{D}\times \ldots\times\tilde{\Omega}^{(d)}_\mathcal{D},
\end{equation}
where $\tilde{\Omega}^{(j)}_\mathcal{D}\subset \mathbb{R}$ depends only on $\mathcal{D}^{(j)}$ -- i.e.,  $\tilde{\Omega}_\mathcal{D}$ has a Cartesian product structure. Additionally, as discussed in Section~\ref{ss:PQC} we know that $\tilde{\Omega}_\mathcal{D} \coloneqq\Omega_{\mathcal{D}}\cup \left(-\Omega_{\mathcal{D}}\right)$, where $\Omega_{\mathcal{D}}\cap\left(-\Omega_{\mathcal{D}}\right) = \{\omega_0\}$. Taken together, we see that 
\begin{align}
|\tilde{\Omega}_\mathcal{D}| &= \prod_{j = 1}^d|\tilde{\Omega}_\mathcal{D}^{(j)}|\label{eq:order_product},\\
|\Omega_{\mathcal{D}}| &= \frac{|\tilde{\Omega}_\mathcal{D}|-1}{2} +1.
\end{align}
Now, let us define $N_j := |\tilde{\Omega}^{(j)}_\mathcal{D}|$, and make the assumption that $N_j$ is independent of $d$\footnote{One can see from the discussion in  Appendix~\ref{app:freq_construction} that $N_j$ depends directly only on $L_j$, the number of encoding gates in $\mathcal{D}^{(j)}$, and the spectra of the encoding Hamiltonians in $\mathcal{D}^{(j)}$. This assumption is therefore justified for all data-encoding strategies in which both $L_j$ and the Hamiltonian spectra are independent of $d$, which is standard practice. One can see Table 1 in Ref.~\cite{Caro_2021} for a detailed list of asymptotic upper bounds on $N_j$ for different encoding strategies.}. Furthermore, let us define $N_{\text{min}} = \min\{N_j\}$. We then have that
\begin{equation}
|\tilde{\Omega}_\mathcal{D}| \geq N_{\text{min}}^d
\end{equation}
i.e.,  that the number of frequencies in $\tilde{\Omega}_\mathcal{D}$, and therefore $\Omega_\mathcal{D}$, scales \textit{exponentially} with respect to $d$. From this we can immediately make our first observation:

\begin{observation}Given that the number of elements in $\Omega_{\mathcal{D}}$ scales exponentially with $d$, one \textit{cannot} efficiently store and sample from arbitrary distributions supported on $\Omega_{\mathcal{D}}$.
\end{observation}
As such, we have to restrict ourselves to \textit{structured distributions}, whose structure facilitates efficient sampling. One such subset of distributions are those which are supported only on a polynomial (in $d$) size subset of $\Omega_{\mathcal{D}}$. Another suitable set of distributions is what we call \textit{product-induced} distributions. Specifically, let $\tilde{p}^{(j)}$ be an arbitrary distribution over $\tilde{\Omega}^{(j)}_\mathcal{D}$, and define the product distribution $\tilde{p}$ over $\tilde{\Omega}_\mathcal{D}$ via 
\begin{equation}\label{eq:product}
\tilde{p}\big(\omega = (\omega_1,\ldots,\omega_d)\big) = \tilde{p}^{(1)}(\omega_1)\times\ldots\times \tilde{p}^{(d)}(\omega_d)
.
\end{equation}
Note that, due to the $d$-independence of $|\tilde{\Omega}^{(j)}_\mathcal{D}|$ we can store and sample from $\tilde{p}^{(j)}$ efficiently, which then allows us to sample from $\tilde{p}$ by simply drawing $\omega_j\sim\tilde{p}^{(j)}$ for all $j\in[d]$ and then outputting $\omega = (\omega_1,\ldots,\omega_d)$. However, it may be the case that $\omega\notin\Omega_{\mathcal{D}}$. As such, the natural thing to do is simply output $\omega$ if $\omega\in\Omega_{\mathcal{D}}$, and if not, output $-\omega$. If one does this, then one samples from the distribution $p$ over $\Omega_{\mathcal{D}}$ defined via
\begin{equation}\label{eq:induced_distributions}
p(\omega) := \begin{cases} \tilde{p}(\omega)\hspace{5em}\text{ if } \omega = \omega_0,\\
\tilde{p}(\omega) + \tilde{p}(-\omega) \hspace{1.2em}\text{ else,}
\end{cases}
\end{equation}
which we refer to as a \textit{product-induced} distribution. We can in fact however go further, and use the Cartesian product structure of $\tilde{\Omega}_\mathcal{D}$ to generalize product-induced distributions to \textit{matrix-product-state-induced} (MPS-induced) distributions. To do this, let us label the elements of $\tilde{\Omega}^{(j)}_\mathcal{D}$ via $\tilde{\Omega}^{(j)}_\mathcal{D} = \{\tilde{\Omega}^{(j)}_{k_j}\}$ for $k_j\in [N_j]$. We can then write any $\omega\in \tilde{\Omega}_\mathcal{D}$ via
$\omega = (\omega^{(1)}_{k_1},\ldots,\omega^{(d)}_{k_d})$, for some indexing $(k_1,\ldots,k_d)$. From this, we see that \textit{any} distribution $\tilde{p}$ over $\tilde{\Omega}_\mathcal{D}$  can be naturally represented as a $d$-tensor - i.e.,  as as a tensor with $d$ legs, where the $j$'th leg is $N_j$ dimensional. Graphically, we have that 
\begin{equation}
\tilde{p}[(\omega^{(1)}_{k_1},\ldots,\omega^{(d)}_{k_d})] = \vcenter{\hbox{\includegraphics[height=0.12\linewidth]{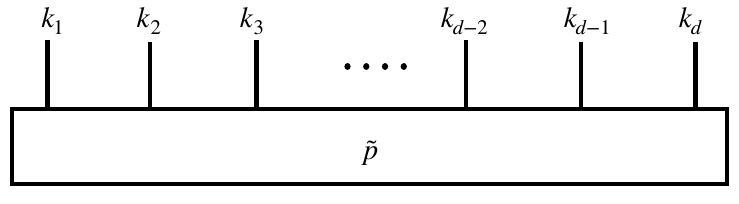}}}
\end{equation}
Now we can consider the subset of distributions which can be represented by a \textit{matrix product state}~\cite{Schollwock_2011}, i.e.,  those distributions for which
\begin{equation}
\tilde{p}[(\omega^{(1)}_{k_1},\ldots,\omega^{(d)}_{k_d})] = \vcenter{\hbox{\includegraphics[height=0.13\linewidth]{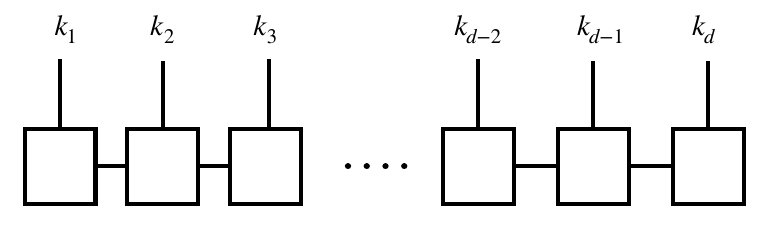}}}
\end{equation}
We refer to distributions which admit such a representation as MPS distributions \cite{PhysRevB.85.165146,Stoudenmire_2010,Expressive}. One can efficiently store such distributions whenever the bond dimension $\chi$ is polynomial in $d$, and as described in Refs.~\cite{PhysRevB.85.165146,Stoudenmire_2010}, one can sample from such distributions with complexity $dN_{\text{max}}\chi^3$. Note that the product distribution in Eq.~\eqref{eq:product} is a special case of an MPS distribution,with $\chi = 1$. Now, given an MPS distribution $\tilde{p}$ over $\tilde{\Omega}_\mathcal{D}$, we define the induced distribution over $\Omega_\mathcal{D}$ via Eq.~\eqref{eq:induced_distributions}.

\textbf{Summary:} In order to \textit{efficiently} implement the RFF procedure for a kernel $K_{(\mathcal{D},\doubleyou)}$, it is necessary that one can sample efficiently from $p_{(\mathcal{D},\doubleyou)}$, the discrete probability distribution associated to the kernel. However, the number of frequency vectors $|\Omega_\mathcal{D}|$ typically scales exponentially in $d$, and as such one \textit{cannot} efficiently store and sample from arbitrary probability distributions over $\Omega_{\mathcal{D}}$. As such, \textit{efficiently} implementing RFF is only possible for the subset of kernels whose associated distributions have a structure which facilitates efficient sampling. Due to the Cartesian product structure of $\tilde{\Omega}_\mathcal{D}$ one such set of distributions (amongst others) are those induced by MPS with polynomial bond dimension.

\subsection{Kernel integral operator for PQC kernels}\label{ss:kern_int_op}

Recall from Theorem~\ref{thm:efficiency_RFF} that 
\begin{equation}
M\geq c_o\sqrt{n}\log\frac{108\sqrt{n}}{\delta}
\end{equation}
frequency samples are sufficient to guarantee, with probability greater than $1-\delta$, an error of at most $\epsilon$ between the output of the RFF procedure and the optimal PQC model. As such, in order to fully understand the complexity of RFF-based regression, it is necessary for us to gain a better understanding of $c_0$, which is given by 
\begin{equation}
c_0 = 9\left(\frac{29}{4} + \frac{4}{\|T_{K_{(\mathcal{D},\doubleyou)}}\|}\right).
\end{equation}
In particular, in order to find the smallest number of sufficient frequency samples, it is necessary for us to obtain an upper bound on $c_0$, which in turn requires a \textit{lower bound} on $\|T_{K_{(\mathcal{D},\doubleyou)}}\|$, the operator norm of the kernel integral operator associated with $K_{(\mathcal{D},\doubleyou)}$. We achieve this with the following lemma, whose proof can be found in Appendix~\ref{app:kern_int_proof}.

\begin{lemma}[Operator norm of $T_{(K_\mathcal{D},\doubleyou)}$] \label{lem:operator_norm}
Let $K_{(\mathcal{D},\doubleyou)}$ be the re-weighted PQC kernel defined via Eq.~\eqref{eq:re-weighted_kernel}, and let $T_{K_{(\mathcal{D},\doubleyou)}}$ be the associated kernel integral operator (as per Definition~\ref{def:kern_int_op}). Assume that (a) the marginal distribution $P_\mathcal{X}$ appearing in the definition of the kernel integral operator is fixed to the uniform distribution, and (b) The frequency set $\Omega_\mathcal{D}$ consists only of integer vectors -- i.e.,  $\Omega_\mathcal{D}\subset \mathbb{Z}^d$. Then, we have that
\begin{align}
  \|T_{K_{(\mathcal{D},\doubleyou)}}\| &= \max_{i \in |\Omega_\mathcal{D}|} \left\{\frac{1}{2}\frac{\doubleyou_i^2}{\|\doubleyou\|^2_2} \right\}\\
  \nonumber
  &=\max_{\omega \in \Omega_\mathcal{D}}\left\{\frac{1}{2}p_{(\mathcal{D},\doubleyou)}(\omega)\right\}.
  \nonumber
\end{align}
\end{lemma}
In light of this, let us again drop the subscript for convenience, and define
\begin{equation}
p_\text{max} \coloneqq \max_{\omega\in\Omega_\mathcal{D}}\left\{p_{(\mathcal{D},\doubleyou)}(\omega)\right\}.
\end{equation} 
With this in hand we can immediately make the following observation:
\begin{observation}\label{obs:prob_decay}
In order to achieve $c_0 = \mathcal{O}\left(\mathrm{poly}(d)\right)$, which is necessary for Theorem~\ref{thm:efficiency_RFF} to imply that ${M=\mathcal{O}(\mathrm{poly}(d))}$ frequency samples are sufficient, we require that
\begin{equation}\label{eq:inv_poly_lower}
p_\text{max} = \Omega\left(\frac{1}{\mathrm{poly}(d)}\right),
\end{equation}
i.e.,  the maximum probability of the probability distribution associated to $K_{(\mathcal{D},\doubleyou)}$ must decay at most inversely polynomially in $d$. This is illustrated in Figure~\ref{fig:concentration}.
\end{observation}
It is important to stress that we have \textit{not} yet established the \textit{necessity} that $p_\text{max}$ decays inversely polynomially for the RFF procedure to be efficient. In particular, we have only established that this is required for the guarantee of Theorem~\ref{thm:efficiency_RFF} to be meaningful. However, we will show shortly, in Section~\ref{ss:lower_bounds}, that Eq.~\eqref{eq:inv_poly_lower} is indeed also necessary, at least in order to obtain a small \textit{average} error.

\begin{figure}[t]
\begin{center}
\includegraphics[width=\linewidth]{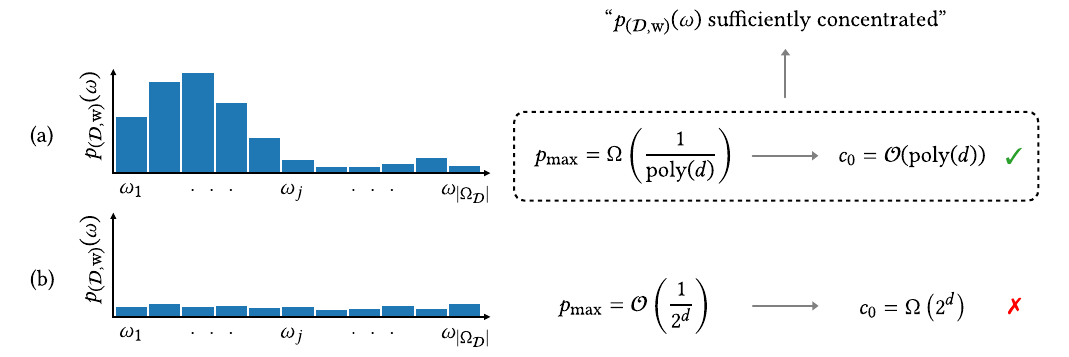} 
\caption{A graphical illustration of Observation~\ref{obs:prob_decay}. In particular, we see that if the re-weighting probability distribution $p_{(\mathcal{D},\doubleyou)}$ is sufficiently concentrated, then $c_0$ will scale polynomially in $d$, which implies via Theorem~\ref{thm:efficiency_RFF} that polynomially many frequency samples $M$ are sufficient to achieve the desired guarantee from RFF-based linear regression.}
\label{fig:concentration}
\end{center}
\end{figure}

In light of Observation~\ref{obs:prob_decay}, we can immediately rule out the meaningful applicability of Theorem~\ref{thm:efficiency_RFF} for kernels with the following associated distributions:

\textbf{The uniform distribution:} As discussed in Section~\ref{ss:implementation}, we have that $|\tilde{\Omega}_\mathcal{D}| \geq N_\text{min}^d$, and therefore, for the uniform distribution over $|\Omega_\mathcal{D}|$ one has that
\begin{equation}
p_\text{max} \leq  \frac{2}{N_\text{min}^d},
\end{equation}
i.e.,  $p_\mathrm{max}$ scales inverse exponentially with $d$.

\textbf{Product-induced distributions:} Consider a probability distribution $p$  over $\Omega_\mathcal{D}$ induced by the product distribution $\tilde{p}$ over $\tilde{\Omega}_\mathcal{D}$, defined as per Eq.~\eqref{eq:product}. We have that
\begin{align}
p_{\text{max}} &\leq 2\tilde{p}_{\text{max}} \hspace{5em} \text{[via Eq.~\eqref{eq:induced_distributions}]}\\
\nonumber
&= 2\prod_{j\in[d]} \tilde{p}^{j}_\text{max} \hspace{3.1em} \text{[via Eq.~\eqref{eq:product}]}\\
\nonumber
&\leq 2\left(\max_{j\in[d]}\left\{\tilde{p}^{j}_\text{max}\right\}\right)^d.
\nonumber
\end{align}
Therefore, whenever $\max_{j\in[d]}\left\{\tilde{p}^{j}_\text{max}\right\} < 1$, there exists some constant $c>1$ such that
\begin{equation}
p_{\text{max}} \leq \frac{2}{c^d}.
\end{equation}

\textbf{Summary:} In order for Theorem~\ref{thm:efficiency_RFF} to be meaningfully applicable -- i.e., to guarantee the efficiency of RFF for approximating variational QML -- one requires that the operator norm of the kernel integral operator decays at most inverse polynomially with respect to $d$, which via Lemma~\ref{lem:operator_norm} requires that the maximum probability of the distribution associated with the kernel decays at most inversely polynomially with $d$. Unfortunately, this rules out any efficiency guarantee, via Theorem~\ref{thm:efficiency_RFF}, for kernels $K_{(\mathcal{D},\doubleyou)}$ whose associated distribution $p_{(\mathcal{D},\doubleyou)}$ is either the uniform distribution or a product-induced distribution (with all component probability distributions non-trivial).

\subsection{RKHS norm for PQC kernels}\label{ss:RKHS_norm}

Recall from Theorem~\ref{thm:efficiency_RFF} that one requires
\begin{equation}
n\geq \max\left\{\frac{c_1^2\log^4\frac{1}{\delta}}{\epsilon^2},n_0\right\}
\end{equation}
data samples, in order for Theorem~\ref{thm:efficiency_RFF} to guarantee, with probability greater than $1-\delta$, an error of at most $\epsilon$ between the output of the RFF procedure and the optimal PQC model. As such, in order to fully understand the complexity of RFF-based regression, it is necessary for us to gain a better understanding of $c_1$, which is given by 
\begin{equation}
c_1 \leq 8(4b + 3C + 2\sqrt{C}),
\end{equation}
where $b$ is set by the regression problem (and we can assume to be constant), and $C$ is an upper bound on the RKHS norm of the optimal PQC model, with respect to the kernel used for RFF -- i.e.,  
\begin{equation}
\|f^*_{(\Theta,\mathcal{D},O)}\|_{K_{(\mathcal{D},\doubleyou)}}\leq C.
\end{equation}
Therefore, we see that in order to determine the smallest number of sufficient data samples, it is necessary for us to obtain a concrete upper bound on the RKHS norm of the optimal PQC function with respect to the kernel $K_{(\mathcal{D},\doubleyou)}$. More specifically, we would like to understand, for which PQC architectures, and for which kernels, one obtains 
\begin{equation}
C = \mathcal{O}\left(\mathrm{poly}(d)\right)
\end{equation}
as for any such kernel and architecture, we can guarantee, via Theorem~\ref{thm:efficiency_RFF}, the sample efficiency of the RFF procedure for approximating the optimal PQC model.

Ideally we would like to obtain results and insights which are \textit{problem-independent} and therefore we focus here not on the optimal PQC function (which requires knowing the solution to the problem) but on the maximum RKHS norm over the entire PQC architecture. More specifically, given an architecture  $(\Theta,\mathcal{D},O)$ we would like to place upper bounds on  
\begin{equation}
C_{(\Theta,\mathcal{D},O)}: = \max\left\{\|f\|_{K_{(\mathcal{D},\doubleyou)}}\,|\, f\in \mathcal{F}_{(\Theta,\mathcal{D},O)}\right\},
\end{equation}
as this clearly provides an upper bound on $\|f^*_{(\Theta,\mathcal{D},O)}\|_{K_{(\mathcal{D},\doubleyou)}}$ for \textit{any} regression problem. 

We start off with an alternative definition of the RKHS norm, which turns out to be much more convenient to work with than the one we have previously encountered.

\begin{lemma}[Alternative definition of RKHS norm -- Adapted from Theorem 4.21 in Ref.~\cite{SVMbook}]\label{lem:altRKHS_norm} Given some kernel $K\colon\mathcal{X}\times\mathcal{X}\rightarrow\mathbb{R}$ defined via 
\begin{equation}
K(x,x') = \langle \phi(x),\phi(x')\rangle,
\end{equation}
for some feature map $\phi\colon\mathcal{X}\rightarrow \mathcal{X}'$, one has that
\begin{equation}\label{eq:altRKHSnorm}
\|f\|_{K} = \inf\, \{\|v\|_2\,|\, v\in \mathcal{X}' \text{ such that } f(\cdot) = \langle v,\phi(\cdot)\rangle \}
\end{equation}
for all $f\in \mathcal{H}_K$.
\end{lemma}
In words, Lemma~\ref{lem:altRKHS_norm} says that the RKHS norm is defined as the \textit{infimum} over the 2-norms of all hyperplanes in feature space which realize $f$ - i.e.,  the infimum over $\|v\|_2$ for all $v$ such that $f(x) = \langle v, \phi(x)\rangle$. We stress that in general functions in the reproducing kernel Hilbert space \textit{do not} have a unique hyper-plane representation with respect to the feature map. However, as detailed in Observation~\ref{obs:uniqueness} below, for PQC feature maps and data-encoding strategies giving rise to integer frequency vectors (such as encoding strategies using only Pauli Hamiltonians~\cite{data_encoding}), the hyperplane representation is indeed \textit{unique}.

\begin{observation}[Hyperplane uniqueness for integer frequencies]\label{obs:uniqueness} Let $\mathcal{D}$ be an encoding strategy for which ${\Omega^{+}_\mathcal{D}\subset\mathbb{Z}^{d}}$. In this case, one has that
\begin{equation}
\left\{1,\cos(\langle\omega_1, x\rangle),\sin(\langle\omega_1 x\rangle),\ldots,  \cos(\langle\omega_{|\Omega^{+}_\mathcal{D}|} ,x\rangle),\sin(\langle\omega_{|\Omega^{+}_\mathcal{D}|}, x\rangle)\right\}
\end{equation}
is a mutually orthogonal set of functions. Therefore for any strictly positive re-weighting $\doubleyou$, and any $u,v \in \mathbb{R}^{|\tilde{\Omega}_{\mathcal{D}}|}$, if ${f(\cdot) = \langle v,\phi_{(\mathcal{D},\doubleyou)}(\cdot)\rangle}$ and $f(\cdot) = \langle u,\phi_{(\mathcal{D},\doubleyou)}(\cdot)\rangle$ then $u=v$. Specifically, there exists only one hyperplane in feature space which realizes $f$. As a consequence one has, via Lemma~\ref{lem:altRKHS_norm}, that if $f(\cdot) = \langle v,\phi_{(\mathcal{D},\doubleyou)}(\cdot)\rangle$ then
\begin{equation}
\|f\|_{K_{(\mathcal{D},\doubleyou)}} = \|v\|_2.
\end{equation}
\end{observation}
Additionally, we note that for any such encoding strategy, the Fourier transform of any PQC function $f\in\mathcal{F}_{(\Theta,\mathcal{D},O)}$ will immediately yield the hyper-plane $v$ such that $f(\cdot) = \langle v, \phi_{(\mathcal{D},\doubleyou)}(\cdot)\rangle$ for the \textit{uniform} weight vector. The hyper-plane representation with respect to any other weight vector (from which the RKHS norm can be calculated) can then be extracted by re-scaling the hyper-plane components appropriately.
With the above insights in hand, we can do some examples to gain intuition into the behaviour of the RKHS norm.

\begin{example}\label{ex:peaked_with_uniform} Given a data-encoding strategy $\mathcal{D}$, and the uniform weight vector 
\begin{equation}
\doubleyou = \frac{1}{\sqrt{|\Omega_{\mathcal{D}}|}}(1,\ldots,1), 
\end{equation}
consider the function $f(x) = \cos(\langle\omega_1, x\rangle)$. In this case one has that $f(\cdot) = \langle v,\phi_{(\mathcal{D},\doubleyou)}(\cdot)\rangle$ with
\begin{equation}
v = \left(0,\sqrt{|\Omega_{\mathcal{D}}|},0,\ldots,0\right),
\end{equation}
and therefore $\|f\|_{K_{(\mathcal{D},\doubleyou)}} \leq \sqrt{|\Omega_{\mathcal{D}}|}$, with an equality in the case of encoding strategies with integer frequency vectors. Note that we obtain the same result for $f(x) = \cos(\langle\omega ,x\rangle)$ and $f(x) = \sin(\langle\omega, x\rangle)$ for any $\omega\in\Omega_\mathcal{D}$. 
\end{example}

\begin{example}\label{ex:peaked_with_peaked} Let us consider the same function as in Example~\ref{ex:peaked_with_uniform} -- i.e.,  $f(x) = \cos(\langle\omega_1, x\rangle)$ -- but this time let us consider the weight vector $\doubleyou = (0,1,0,\ldots,0)$. In this case one has $f(\cdot) = \langle v,\phi_{(\mathcal{D},\doubleyou)}(\cdot)\rangle$ with
\begin{equation}
v = \left(0,1,0\ldots,0\right),
\end{equation}
and therefore $\|f\|_{K_{(\mathcal{D},\doubleyou)}} \leq 1$, with an equality in the case of encoding strategies with integer frequency vectors. We again get the same result for  $f(x) = \cos(\langle\omega, x\rangle)$ and $f(x) = \sin(\langle\omega, x\rangle)$ for any $\omega\in\Omega_\mathcal{D}$, if one uses the weight vector with $\doubleyou_\omega = 1$.
\end{example}

\begin{example}\label{ex:uniform_with_uniform} Given a data-encoding strategy $\mathcal{D}$, and the uniform weight vector 
\begin{equation}
\doubleyou = \frac{1}{\sqrt{|\Omega_{\mathcal{D}}|}}(1,\ldots,1), 
\end{equation}
consider the function
\begin{equation}
f(x) = \frac{1}{|\Omega_\mathcal{D}|}\sum_{\omega\in\Omega_\mathcal{D}}\cos(\langle\omega, x\rangle).
\end{equation}
In this case one has 
$f(\cdot) = \langle v,\phi_{(\mathcal{D},\doubleyou)}(\cdot)\rangle$ with
\begin{equation}
v= \frac{1}{\sqrt{|\Omega_\mathcal{D}| }}\left(1,1,0,1,0\ldots,1,0\right),
\end{equation}
and, therefore, $\|f\|_{K_{(\mathcal{D},\doubleyou)}} \leq 1$, with an equality in the case of encoding strategies with integer frequency vectors.
\end{example}
Given these examples, we can extract the following important observations:

\begin{figure}[t]
\begin{center}
\includegraphics[width=\textwidth]{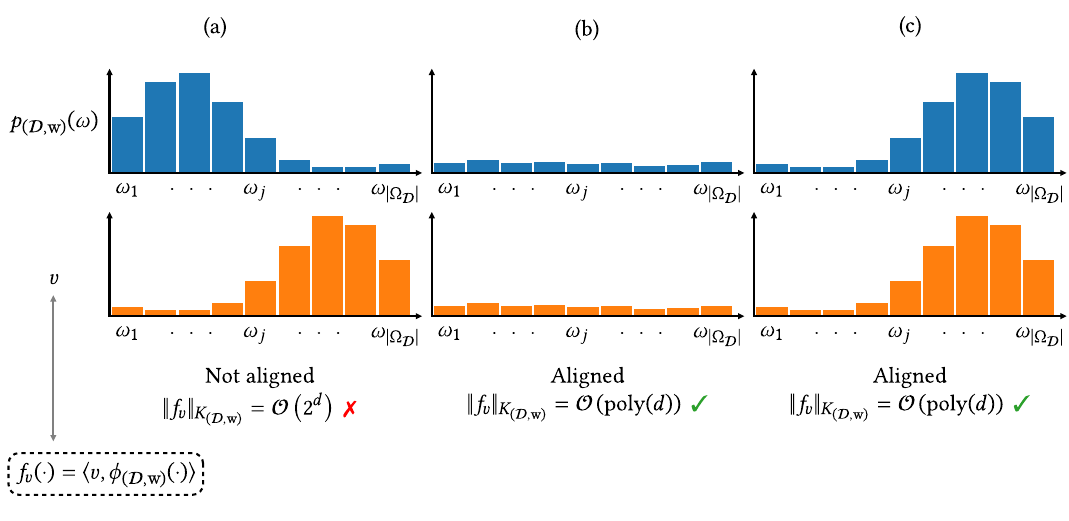} 
\caption{A graphical illustration of the conditions necessary for the RKHS norm of a function to scale polynomially in $d$. At a high level, we see that for a function $f_v(\cdot) = \langle v,\phi_{(\mathcal{D},\doubleyou)}(\cdot)\rangle$, the hyperplane vector $v$ (equivalently frequency spectrum of $f$), needs to be sufficiently well-aligned with the re-weighting vector $\doubleyou$, which determines the kernel $K_{(\mathcal{D},\doubleyou)}$ with respect to which the RKHS norm is taken.}
\label{fig:RKHSnorm}
\end{center}
\end{figure}

\begin{enumerate}
\item As per Example~\ref{ex:peaked_with_uniform}, there exist functions, and re-weighted PQC kernels, for which the RKHS norm of the function scales  with the number of frequencies in $\Omega_\mathcal{D}$, and therefore exponentially in $d$. As such, we \textit{cannot} hope to place a universally applicable (architecture independent) polynomial (in $d$) upper bound on $C_{(\Theta,\mathcal{D},O)}$. On the contrary, as per Examples~\ref{ex:peaked_with_peaked} and~\ref{ex:uniform_with_uniform} there do exist functions and reweightings for which the RKHS norm is \textit{constant}. Therefore, while we cannot hope to place an architecture-independent upper bound on the RKHS norm, it may be the case that there exist specific circuit architectures and kernel re-weightings for which $C_{(\Theta,\mathcal{D},O)}$ is upper bounded by a polynomial in $d$. Note that this can be interpreted as an \textit{expressivity constraint} on $(\Theta,\mathcal{D},O)$, as the more expressive an architecture is, the more likely it contains a function with large RKHS norm (with this likelihood becoming a certainty for the case of universal architectures).
\item By comparing Examples~\ref{ex:peaked_with_uniform} and~\ref{ex:peaked_with_peaked} we see that, as expected, the RKHS norm of  a function depends strongly on the reweighting which defines the kernel. In particular, the same function can have a very different RKHS norm with respect to different feature map re-weightings.
\item By looking at all examples together, we see that informally, what seems to determine the RKHS norm of a given function $f_v$ is the ``alignment'' between (a) the frequency distribution of the function $f_v$, i.e.,  the components of the vector $v$, and (b) the re-weighting vector $\doubleyou$ (or alternatively, the probability distribution $p_{(\mathcal{D},\doubleyou)}$). In particular, in Example~\ref{ex:peaked_with_uniform}, the frequency representation of the function $f$ is peaked on a single frequency, whereas the probability distribution $p_{(\mathcal{D},\doubleyou)}$ is uniform over all frequencies. In this example, $v$ and $p_{(\mathcal{D},\doubleyou)}$ are \textit{non-aligned}, and we find that the RKHS norm of $f$ with respect to $K_{(\mathcal{D},\doubleyou)}$ scales exponentially in $D$. On the contrary, in both Examples~\ref{ex:peaked_with_peaked} and~\ref{ex:uniform_with_uniform} we have that the frequency representation of $f$ is well aligned with the probability distribution $p$, and we find that we can place a \textit{constant} upper bound on the RKHS norm of the function. This is illustrated in Figure~\ref{fig:RKHSnorm}.
\item At an intuitive level, one should expect the ``alignment'' between the frequency representation of a target function and the probability distribution associated with the kernel to play a role in the complexity of RFF. Informally, in order to learn an approximation of the function $f_v$ via RFF, when constructing the approximate kernel via frequency sampling we need to sample frequencies present in $v$. Therefore, if the distribution is supported mainly on frequencies \textit{not} present in $v$, we cannot hope to achieve a good approximation via RFF. On the contrary, if the distribution is supported on frequencies present in $v$, with the correct weighting, then we can hope to approximate $f_v$ using our approximate kernel. In this sense, our informal observation that the RKHS norm depends on the alignment of target function with kernel probability distribution squares well with our intuitive understanding of RFF-based linear regression. We will make this intuition much more precise in Section~\ref{ss:lower_bounds}.
\end{enumerate}

\textbf{Summary:} In order for the statement of Theorem~\ref{thm:efficiency_RFF} to imply that a polynomial number of data samples is sufficient, we require that $\|f^*_{(\Theta,\mathcal{D},O)}\|_{K_{(\mathcal{D},\doubleyou)}}$, the RKHS norm of the optimal PQC function, scales polynomially with respect to $d$. Unfortunately, in the worst case, $\|f^*_{(\Theta,\mathcal{D},O)}\|_{K_{(\mathcal{D},\doubleyou)}}$ can scale exponentially with respect to $d$, and therefore we cannot hope for efficient dequantization of variational QML via RFF for \textit{all} possible circuit architectures. However, given a specific circuit architecture, and a re-weighting which leads to a distribution $p_{(\mathcal{D},\doubleyou)}$ with an efficient sampling algorithm, it may be the case that $C_{(\Theta,\mathcal{D},O)}$ scales polynomially in $d$, in which case Theorem~\ref{thm:efficiency_RFF} yields a meaningful sample complexity for the RFF dequantization of optimization of $(\Theta,\mathcal{D},O)$ for \textit{any} regression problem $P$. Unfortunately however, it seems unlikely that an expressive circuit architecture will not contain \textit{any} functions with large RKHS norm. Ultimately though, all that is required by Theorem~\ref{thm:efficiency_RFF} is that the RKHS norm of the \textit{optimal} PQC function scales polynomially in $d$, and this may be the case even when $C_{(\Theta,\mathcal{D},O)}$ scales superpolynomially. Unfortunately, given a regression problem $P$ it is not clear how to assess the RKHS norm of the optimal PQC function without knowing this function in advance, which seems to require running the PQC optimization.

\subsection{Lower bounds for RFF efficiency}\label{ss:lower_bounds}

By this point we have seen that the following is \textit{sufficient} for Theorem~\ref{thm:efficiency_RFF} to imply the efficient dequantization of variational QML via RFF-based linear regression:

\newpage

\begin{enumerate}
\item We need to be able to efficiently sample from $p_{(\mathcal{D},\doubleyou)}$.
\item The distribution $p_{(\mathcal{D},\doubleyou)}$ needs to be sufficiently concentrated. In particular, we need $p_{\text{max}} = \Omega(1/\mathrm{poly}(d))$ in order to place a sufficiently strong lower bound on the operator norm of the kernel integral operator. 
\item The frequency representation of the optimal PQC function needs to be ``well aligned'' with the probability distribution $p_{(\mathcal{D},\doubleyou)}$. This is required to ensure a sufficiently strong upper bound on the RKHS norm of the optimal PQC function.
\end{enumerate}

It is clear that point 1 above is a \textit{necessary} criterion for efficient dequantization via RFF. However, it is less clear to which extent points 2 and 3 are necessary. In particular, it could be the case that the bounds provided by Theorem~\ref{thm:efficiency_RFF} are not tight, and that efficient PQC dequantization via RFF is possible even when not guaranteed by Theorem~\ref{thm:efficiency_RFF}. In this section we address this issue to some extent, by proving \textit{lower bounds} on the complexity of RFF which show that both points 2 and 3 are also necessary conditions in some sense, when using certain encoding strategies.
To be more specific, Theorem~\ref{thm:efficiency_RFF} provides sufficient conditions for the output of RFF-based linear regression to be (a) with high probability, (b) no more than $\epsilon$ worse than the optimal PQC function with respect to true risk. In this section we show necessary conditions -- when using encoding strategies with integer frequencies -- to ensure that (a) the expected, (b) $L^2$-norm difference, is small between the output of RFF-based linear regression and the optimal PQC function. As such, the necessary conditions we provide here are different from the sufficient conditions of Theorem~\ref{thm:efficiency_RFF} in the following ways:
\begin{enumerate}
    \item They ensure that the $L^2$-norm difference is small between the output of RFF-regression and the optimal PQC function, as opposed to the true risk difference.
    \item They ensure the above difference is small \textit{in expectation}, as opposed to with high probability.
    \item They apply only to circuit architectures using encoding strategies with integer-valued frequency vectors (such as those using Pauli word Hamiltonians).
\end{enumerate}

To make this more precise, we need to introduce some additional notation. We consider a data encoding strategy, giving rise to an integer-valued frequency set $\tilde{\Omega}_\mathcal{D}$, as well as a weight vector $\doubleyou$, giving rise to the sampling distribution $p_{(\mathcal{D},\doubleyou)}$, which we abbreviate as $p$. Now, given a regression problem, we define the following notions:
\begin{enumerate}
\item Let $f^*_{(\Theta,\mathcal{D},O)}$ represent the optimal PQC model. In this section, for convenience we abbreviate this as $f^*$.
\item As per Section~\ref{ss:RFF} and the description of Algorithm~\ref{alg:RFF} in Section~\ref{s:PQCvsRFF}, we consider running RFF regression by sampling $M$ frequencies from the distribution $\pi = p\times \mu$. Let $\vec{\omega}=(\omega_1,\ldots,\omega_M)\in \Omega^M_\mathcal{D}$ be the random variable of $M$ frequencies sampled from $p$, and let $g_{\vec{\omega}}$ be the output of linear regression using these frequencies to approximate the kernel $K_{(\mathcal{D},\doubleyou)}$.
\end{enumerate}

As discussed above, in this section we are concerned with lower bounding the \textit{expected} $L^2$-norm of the difference between the optimal PQC function and the output of the RFF procedure, with respect to multiple runs of RFF-based linear regression. In particular, we want to place lower bounds on the quantity
\begin{align}
\hat{\epsilon} \coloneqq & \underset{\vec{\omega}\sim p^M}{\mathbb{E}} \|f^* - g_{\vec{\omega}}\|^2_2 \\
=&\sum_{\vec{\omega}\in \Omega^M_\mathcal{D}}\|f^* - g_{\vec{\omega}}\|_2^2\xi(\vec{\omega}),
\nonumber
\end{align}
where $\xi = p^M$, i.e.,  $\xi(\vec{\omega}) = p(\omega_1)\times\ldots\times p(\omega_M)$. In order to lower bound $\hat{\epsilon}$, recall from Section~\ref{ss:PQC} that $f^{*}$ can be written as
\begin{equation}
f^*(x) = \sum_{\omega\in\tilde{\Omega}_\mathcal{D}}\hat{f}^*(\omega)e^{i\langle\omega, x\rangle}.
\end{equation}
We abuse notation slightly and use the notation $\hat{f}^*$ to denote the vector with entries $\hat{f}^*(\omega)$. Finally, we denote by $p_\text{max}$ the maximum probability in $p$. With this in hand, we have the following lemma (whose proof can be found in Appendix~\ref{app:lower}):

\begin{lemma}[Lower bound on average relative error]\label{lem:lower_bound} Given any encoding strategy $\mathcal{D}$ for which all $\omega\in\Omega_\mathcal{D}$ have only integer-valued components, the expected $L^2$-norm of the difference between the optimal PQC function and the output of RFF-based linear regression can be lower bounded as
\begin{align}
    \hat{\epsilon} 
    &\geq (2\pi)^d\|\hat{f}^*\|_2^2 -  (2\pi)^d2M \sum_{\omega \in \Omega_\mathcal{D}} |\hat{f}^*(\omega)|^2 p(\omega)\label{eq:lower_overlap}
    \\
    \nonumber
    &\geq (2\pi)^d\|\hat{f}^*\|_2^2 -  (2\pi)^d2M \sum_{\omega \in \Omega_\mathcal{D}} |\hat{f}^*(\omega)|^2 p_{\max}\\
    &= \|f^*\|_2^2\left(1-2Mp_{\max}\right).\label{eq:lower_pmax}
\end{align}
\end{lemma}
Using this Lemma, we can now see that, at least for some class of regression problems, both concentration of the probability distribution $p$, and ``alignment'' of the frequency representation of the optimal function with $p$, are \textit{necessary} conditions to achieve a small \textit{expected} relative error $\hat{\epsilon}$, when using encoding strategies with integer-valued frequencies.

\textbf{Concentration of $p$:} To this end, note that we can rewrite Eq.~\eqref{eq:lower_pmax} as
\begin{equation}\label{eq:sample_lower}
M \geq \frac{1}{2p_\text{max}}\left(1-\frac{\hat{\epsilon}}{\|f^*\|_2^2}\right).
\end{equation}
Given this, we see that when $(1-\hat{\epsilon}/\|f^*\|_2^2) = \Omega(1)$ then one requires ${M= \Omega(1/p_\text{max})}$ frequency samples to achieve expected relative error $\hat{\epsilon}$. We stress however that the condition $(1-\hat{\epsilon}/\|f^*\|_2^2) = \Omega(1)$ will \textit{not} always be satisfied. More specifically, one requires that $\|f^*\|_2^2 \geq c + \hat{\epsilon}$ for all $d$, for some \textit{constant} $c$. In particular, we note that the bound of Eq.~\eqref{eq:sample_lower} is vacuous whenever $\|f^*\|^2_2 \leq \hat{\epsilon}$.
However when $(1-\hat{\epsilon}/\|f^*\|_2^2) = \Omega(1)$ is satisfied, the RFF procedure \textit{cannot} be efficient whenever $p_\text{max}$ is a negligible function -- i.e.,  decays faster than any inverse polynomial. More specifically, when $p$ is not sufficiently concentrated, one will require super-polynomially many samples $M$.  Recall from Section~\ref{ss:kern_int_op} that for all product-induced distributions, including the uniform distribution, $p_\text{max}$ is a negligible function of $d$ - and therefore we \textit{cannot} achieve efficient RFF dequantization of variational QML via any re-weighting giving rise to such a distribution, when $(1-\hat{\epsilon}/\|f^*\|_2^2) = \Omega(1)$ is satisfied. Complimenting the lower bound, we recall that as discussed in Section~\ref{ss:kern_int_op}, whenever $p_{\text{max}} = \Omega(1/\mathrm{poly}(d))$ -- i.e., when $p$ is sufficiently concentrated -- then one \textit{can} apply Theorem~\ref{thm:efficiency_RFF} to place a polynomial upper bound on $M$.

\textbf{Alignment of $\hat{f}^*$ and $p$:} Note that we can rewrite Eq.~\eqref{eq:lower_overlap} as
\begin{equation}
M \geq \frac{1}{2\sum_{\omega \in \Omega_\mathcal{D}} |\hat{f}^*(\omega)|^2 p(\omega)} \left(\|\hat{f}^*\|^2_2 - \frac{\hat{\epsilon}}{(2\pi)^d}\right).
\end{equation}
Therefore, whenever $\|\hat{f}^*\|^2_2 - \hat{\epsilon}/(2\pi)^d = \Omega(1)$ one has that 
\begin{equation}
M = \Omega\left(\frac{1}{\sum_{\omega \in \Omega_\mathcal{D}} |\hat{f}^*(\omega)|^2 p(\omega)}\right)
\end{equation}
where 
\begin{equation}
\sum_{\omega \in \Omega_\mathcal{D}} |\hat{f}^*(\omega)|^2 p(\omega)
\end{equation}
is the inner product between the frequency vector $\hat{f}^*$ and the probability distribution $p$, which we interpret as the ``alignment'' between the frequency representation and the sampling distribution. We therefore see that the smaller this overlap, the larger number of frequencies are required to achieve a given relative error. Again, we therefore see that ``large alignment'' between $\hat{f}^*$ and the probability distribution $p$ is a necessary condition to achieve a smaller expected relative error.

\section{Discussion and conclusions}\label{s:conclusions}

In this work, we have provided a detailed analysis of classical linear regression with random Fourier features, using re-weighted PQC kernels, as a method for the dequantization of PQC based regression. Intuitively, as discussed in Section~\ref{s:PQCvsRFF}, this method is motivated by the fact that it optimizes over a natural extension of the same function space used by PQC models -- i.e.,  the method has to some extent an  \textit{inductive bias} which is comparable to that of PQC regression. At a very high level, given a PQC architecture $(\Theta,\mathcal{D},O)$ and a regression problem $P\sim \mathcal{X}\times \mathbb{R}$, the method consists of:

\begin{enumerate}
\item Choosing a re-weighting $\doubleyou$ of the PQC feature map, or equivalently, choosing a distribution $p$ over frequencies appearing in $\Omega_\mathcal{D}$.
\item Sampling $M$ frequencies from $p$, and using them to construct an approximation of the PQC feature map.
\item Being given, or sampling, $n$ training data points from $P$ and running regularized linear regression with the approximate feature map.
\end{enumerate}
The main intuitive take-aways from this work are that in order for this dequantization procedure to output a good approximation to the optimal PQC function $f^*$ one needs to sample frequencies present in the multivariate Fourier decomposition of $f^*$, and in order to be able to do this \textit{efficiently} only a polynomial number of these frequencies should be sufficient for a good approximation.

To be more precise, we know that Step 3 above has time and space complexity $\mathcal{O}(nM)$ and $\mathcal{O}(nM^2 + M^3)$, respectively. Given this, the question we have addressed is whether it is possible to use some number of frequencies $M=\mathcal{O}(\mathrm{poly(d))}$ and some number of data samples $n = \mathcal{O}(\mathrm{poly}(d))$, such that, with high probability, the output of classical RFF-based linear regression is guaranteed to achieve a true risk which is no more than $\epsilon$ worse than the output of PQC based optimization. In other words, whether we can \textit{efficiently} dequantize PQC regression via architecture $(\Theta,\mathcal{D},O)$ for regression problem $P$. In order to answer this question we have been interested in obtaining necessary and sufficient conditions on $n$ and $M$ -- in terms of properties of the circuit architecture $(\Theta,\mathcal{D},O)$, re-weighting $\doubleyou$ and regression problem $P$ -- in order to achieve the desired guarantee. To this end we have seen, via Theorem 1 and the subsequent analysis, that RFF-based linear regression with re-weighting $\doubleyou$ provides an efficient dequantization of PQC regression over the circuit architecture $(\Theta,\mathcal{D},O)$, for regression problem $P$, if: 
\newpage
\begin{enumerate}
\item The re-weighting distribution $p_{(\mathcal{D},\doubleyou)}$ is sufficiently concentrated. In particular,  $p_\text{max}$ should decay at most inversely polynomially in $d$.
\item The re-weighting distribution $p_{(\mathcal{D},\doubleyou)}$ is sufficiently ``well-aligned" with the optimal PQC function $f^*$ for the regression problem $P$. Technically, the RKHS norm of the optimal PQC function $\|f^*_{(\Theta,\mathcal{D},O)}\|_{K_{(\mathcal{D},\doubleyou)}}$ should scale polynomially in~{$d$}.
\item There exists an efficient algorithm to sample from $p_{(\mathcal{D},\doubleyou)}$, given as input only the data encoding strategy $\mathcal{D}$.
\end{enumerate}
We stress again the intuitive perspective on the above conditions: In order for the dequantization procedure to output a good approximation to the optimal PQC function $f^*$ we need to sample frequencies present in $f^*$ (condition 2), and in order to be able to do this efficiently only a polynomial number of these frequencies should be sufficient for a good approximation (condition 1), and the sampling procedure itself should be efficient (condition 3). Additionally, we note that in order to satisfy conditions 1 and 2 above, it is necessary that the frequency representation of the optimal PQC function $f^*$ is sufficiently concentrated. If this is not the case, then it is impossible for the re-weighting distribution $p_{(\mathcal{D},\doubleyou)}$ to be both sufficiently concentrated \textit{and} well-aligned with $f^*$. Given this, we have summarized the sufficient conditions given above graphically in Figure.~\ref{fig:RKHS_conditions}. On the other hand, we have seen, via Lemma~\ref{lem:lower_bound}, that these conditions are also \textit{necessary} in a certain sense -- i.e.,  that if they are not satisfied then -- at least for integer-valued encoding strategies and certain regression problems --  RFF-based linear regression will \textit{not} provide an efficient dequantization of PQC regression (on average). 

\begin{figure}
\begin{center}
\includegraphics[width=\textwidth]{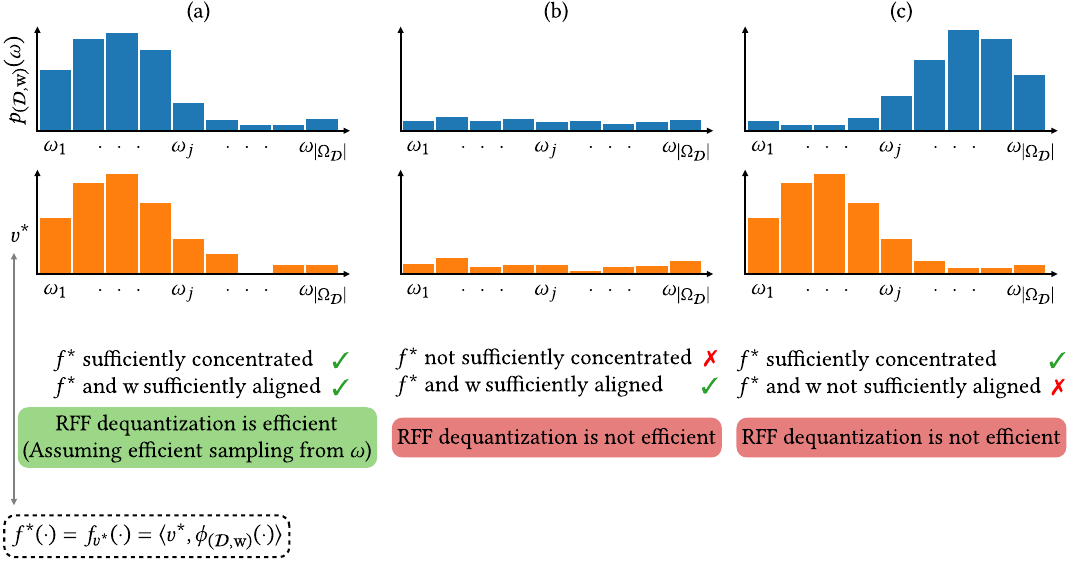}
\caption{A graphical illustration of the conditions sufficient for RFF-based linear regression with re-weighting $\doubleyou$ to provide an efficient dequantization of PQC regression via circuit architecture $(\Theta,\mathcal{D},O)$, for regression problem $P$. One requires that the optimal PQC function $f^*$ for $P$ is both sufficiently concentrated and sufficiently well-aligned with the re-weighting distribution $p_{(\mathcal{D},\doubleyou)}$.}\label{fig:RKHS_conditions}
\end{center}
\end{figure}

Given these insights, perhaps the most interesting questions are the following:
\newpage
 \begin{enumerate}
 \item Given a PQC architecture $(\Theta,\mathcal{D},O)$, as well as (samples from) a regression problem $P$, how do we evaluate whether PQC regression over $(\Theta,\mathcal{D},O)$ can be efficiently dequantized via RFF-based linear regression?
 \item If PQC regression can be efficiently dequantized, how do we identify an appropriate re-weighting $\doubleyou$?
 \end{enumerate}
 
Figure~\ref{fig:flowchat} provides a methodology for answering this question, using the results of this work as a guide. We elaborate on this below.

The first step is to ask whether variational QML with the PQC architecture $(\Theta,\mathcal{D},O)$ can be dequantized via RFF-based linear regression, \textit{irrespective} of the regression problem of interest. More specifically, for any regression problem $P$, there will exist some optimal PQC function $f^*$. Previously, we have focused our analysis on this specific function $f^*$, and noted that it needs to be both sufficiently concentrated and well aligned with a re-weighting whose associated distribution is efficient to sample from. However, if there exists a sufficiently concentrated re-weighting~$\doubleyou$, which can be efficiently sampled from, and which is well aligned with \textit{all functions} $f\in\mathcal{F}_{(\Theta,\mathcal{D},O)}$, then using $\doubleyou$ for the re-weighting will ensure that the sufficient conditions for dequantization will be satisfied \textit{for any} regression problem~$P$. More specifically, no matter which $f\in\mathcal{F}_{(\Theta,\mathcal{D},O)}$ is the optimal PQC function $f^*$, it will be well aligned with~$\doubleyou$, which is sufficiently concentrated. From a more technical perspective, this step of the methodology is equivalent to asking whether there exists a sufficiently concentrated re-weighting $\doubleyou$, which can be efficiently sampled from, such that for all $f\in\mathcal{F}_{(\Theta,\mathcal{D},O)}$ the RKHS norm $\|f_{(\Theta,\mathcal{D},O)}\|_{K_{(\mathcal{D},\doubleyou)}}$ scales polynomially in~$d$. As illustrated in Figure.~\ref{fig:flowchat}, if the answer to this question is ``Yes", then we know that dequantization of $(\Theta,\mathcal{D},O)$ is possible via re-weighting $\doubleyou$ for \textit{any} regression problem $P$. If the answer is "No" (or the question is too hard to answer in practice) then we proceed to examine the specific regression problem of interest.

Before continuing it is worth commenting on a few aspects of the above discussion. Firstly, in order for all ${f\in\mathcal{F}_{(\Theta,\mathcal{D},O)}}$ to be well-aligned with a single sufficiently concentrated re-weighting $\doubleyou$, they should all be sufficiently concentrated and mutually well-aligned. In principle, we can use this insight to guide PQC architecture design: in order to avoid being dequantizable via RFF for any regression problem, a PQC architecture should contain functions whose frequency spectra are either not sufficiently concentrated, or not mutually well-aligned. While useful in principle, we stress that in practice this may be difficult to assess. In particular, determining this property of $\mathcal{F}_{(\Theta,\mathcal{D},O)}$ may in the worst-case require analyzing the frequency spectrum of \textit{all} functions $f\in\mathcal{F}_{(\Theta,\mathcal{D},O)}$. We elaborate more on the challenge and potential of using these insights for PQC architecture design in Section~\ref{s:future}.

While it may be the case that for certain restricted circuit architectures problem-independent dequantization is possible, we expect that for any sufficiently expressive circuit architecture this will not be the case, and therefore it will be necessary to perform an evaluation which depends on the regression problem of interest. To this end, as discussed before, we need to assess whether the \textit{optimal} PQC function $f^*$ for the regression problem is well-aligned with a sufficiently concentrated re-weighting distribution, from which one can efficiently sample. Unfortunately, it is not a-priori clear how to do this without first solving the regression problem and identifying~$f^*$!

\begin{figure}[p]
\begin{center}
\includegraphics[width=\textwidth]{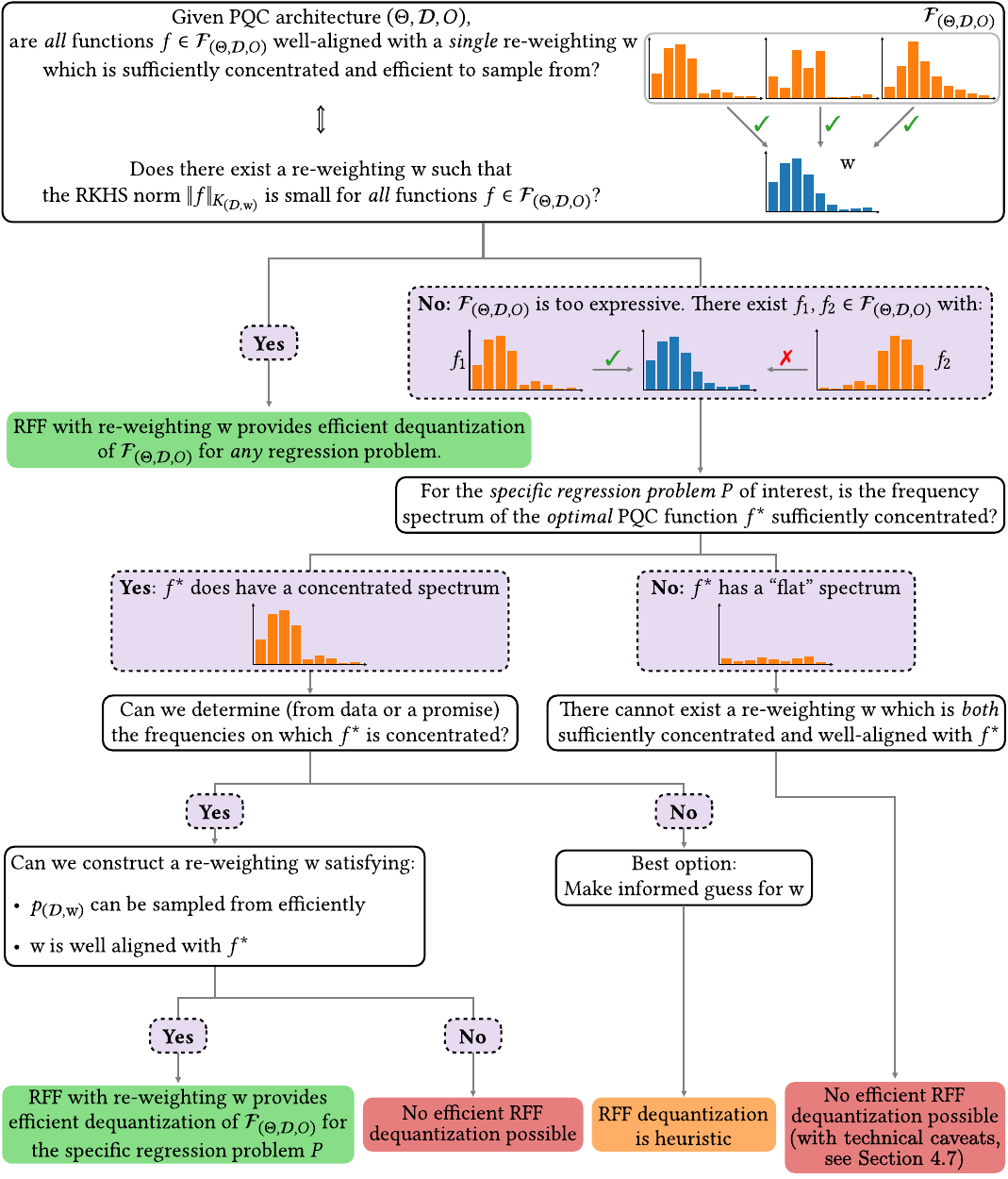}
\caption{A methodology for determining whether, and via which re-weighting, linear regression via RFF can provide an efficient means for dequantizing PQC regression over circuit architecture $(\Theta,\mathcal{D},O)$, for regression problem $P$?}\label{fig:flowchat}
\end{center}
\end{figure}

However, as illustrated in Figure~\ref{fig:flowchat}, we note that a natural way to approach this problem is via the following sequence of questions: Firstly, given samples (i.e., data) from the regression problem, is the frequency spectrum of the optimal PQC function is sufficiently concentrated? If yes, can we identify the frequencies on which the support is concentrated? In principle, it is possible that answering these two questions is easier than learning the function itself, which would also require learning the actual frequency co-efficients. This is in line with the observation underlying the field of property testing, that often \textit{testing} properties of a function can be done more efficiently than \textit{learning} the function~\cite{goldreich2017introduction}. As discussed in Section~\ref{s:future} below, the development of techniques for answering these two property-testing questions is therefore a natural and well-defined direction of future research. Moreover, it may be the case that for certain regression problems one has a \textit{structural promise} which allows one to answer these questions. Indeed, even a partial answer may allow one to make a well-informed \textit{guess} for the re-weighting distribution, which leads to a well-informed \textit{heuristic} method for dequantization. Additionally, as also illustrated in Figure.~\ref{fig:flowchat}, if we are able to identify that the optimal PQC function $f^*$ is \textit{not} sufficiently concentrated, then we can immediately conclude that efficient dequantization via RFF is \textit{not} possible (provided that the technical caveats discussed in detail in Section~\ref{ss:lower_bounds}, on the encoding strategy and 2-norm of $f^*$,  are satisfied). As discussed in Section~\ref{s:future} below, this therefore provides a potential tool for identifying candidate problems for ``quantum advantage" -- i.e., problems for which efficient dequantization via RFF-based linear regression is \textit{not} possible.

Finally, as illustrated in Figure.~\ref{fig:flowchat}, we note that even if we are able to identify the frequencies on which the optimal PQC function $f^*$ is concentrated, efficient RFF dequantization still requires the construction of a re-weighting distribution supported on these frequencies, from which one can efficiently sample. Once again, it is not immediately a-priori clear how one should construct such distributions. As such, the development of methods for the identification and design of such distributions is once again a natural avenue for future research.

\section{Future directions}\label{s:future}
 
Given the discussion above, the following are natural avenues for future research:
 
\textbf{Identification of problems admitting potential quantum advantage:} Can we identify a class of scientifically, industrially or socially relevant regression problems, whose regression functions are well aligned with anti-concentrated distributions over exponentially large frequency sets, and are therefore good candidates for quantum advantage via variational QML?

\textbf{Testing frequency concentration and support:} A key step of the methodology outlined in Figure~\ref{fig:flowchat} involves determining whether the optimal PQC function, for a specific regression problem, has a sufficiently concentrated frequency spectrum. As such, the development of efficient algorithms for \textit{testing} this property of the optimal PQC function from a specific architecture, from samples of the regression problem at hand, would facilitate the practical assessment of whether PQC regression via a specific architecture can be dequantized, for a specific regression problem.

\textbf{Design of suitable sampling distributions:} We have seen that a \textit{necessary} condition for efficient dequantization via RFF is that the distribution $p$ is both efficiently sampleable and sufficiently concentrated. This immediately \textit{rules out} a large class of natural distributions - namely the uniform distribution and product-induced distributions with non-trivial components. Given this, in order for RFF-based dequantization to be useful, it is important to identify and motivate suitable sampling distributions. We note that ideally one would choose the distribution based on knowledge of the regression function of the problem $P$, however in practice it is more likely that one would first choose a distribution $p$, which will then determine the class of problems for which RFF will be an efficient dequantization method -- namely those problems whose regression function is well aligned with $p$.

\textbf{PQC architecture design guided by RKHS norm:} As we have discussed, for any circuit architecture for which \textit{all} expressible functions are well aligned with a suitable distribution $p$, one \textit{cannot} obtain a quantum advantage, as RFF-based linear regression will provide an efficient dequantization method. Given this, it is of interest to investigate circuit architectures from this perspective, to understand which architectures might facilitate a quantum advantage, and which architectures are prone to dequantization. Unfortunately, we note that gaining \textit{analytic} insight into which hyperplanes (i.e., frequency representations) are expressible by a given PQC architecture is a hard problem, which has so far resisted progress. However, as observed in Section~\ref{ss:RKHS_norm}, for data-encoding strategies with integer frequencies -- i.e., all encoding strategies using Pauli word Hamiltonians -- one can in principle numerically evaluate the RKHS norm of a given PQC function via its Fourier transform. Of course, in practice this will be limited by the fact that the frequency set typically scales exponentially with the dimensionality of the data. However, the hope is that one may be able to extract meaningful and useful insights by studying both smaller PQC model classes, for which the RKHS norm calculations are tractable numerically, and more structured PQC model classes for which analytic calculations may be possible.

\textbf{Effect of noise on RKHS norm of PQC architectures:} As we have pointed out, for any sufficiently expressive circuit architecture we expect the worst case RKHS norm to scale superpolynomially -- i.e.,  that there exist PQC functions whose frequency representation is aligned with an anti-concentrated distribution over frequencies. It would, however, be of interest to understand the effect of noise on architectures realizing such functions. In particular, it could be the case that realistic circuit noise causes a concentration of the frequencies which are expressible by a PQC architecture, and therefore facilitates dequantization via RFF-based linear regression.

\textbf{Extension to classification problems:} The analysis we have performed here has focused on \textit{regression} problems. Extending this analysis to classification problems would be both natural and interesting.

\textbf{Improved RFF methods for sparse data:} In this work, we have provided an analysis of ``standard'' RFF-based linear regression, for regression problems with no promised structure. However, when one is promised that the distribution $P$ has some particular structure, then one can devise variants of RFF with improved efficiency guarantees. One such example we have already seen -- namely, if one can guarantee that the regression function has a frequency representation supported on a subset of possible frequencies, then on can design the sampling distribution appropriately, which leads to improved RFF efficiencies~\cite{genRFF}. In a similar vein, it is known that when one has a promise on the sparsity of the vectors in the support of $P$, then one can devise variants of RFF with improved efficiency~\cite{rick2016random}. As this is a natural promise for application-relevant distributions, understanding the potential and limitations of ``Sparse-RFF" as a dequantization method is an interesting research direction.

\textbf{PQC dequantization without RFF:}  Here we have discussed only \textit{one} potential method for the dequantization of variational QML. As we have noted in Section~\ref{s:PQCvsRFF}, recently Ref.~\cite{shin2023analyzing} has noted that to each PQC one can associate a feature map for which the associated kernel can be evaluated efficiently classically \textit{without} requiring any approximations. As such, understanding the extent to which one can place relative error guarantees on linear regression using such a kernel is a natural avenue for investigation. Additionally, as mentioned in the introduction, a variety of recent works have shown that PQCs can be efficiently simulated, \textit{in an average case sense}, in the presence of certain types of circuit noise~\cite{fontana2023classical,shao2023simulating}. Again, understanding the extent to which this allows one to classically emulate noisy variational QML is another natural approach to dequantization of \textit{realistic} variational QML. Finally, quite recently Ref.~\cite{jerbi2023shadows} has shown that one can sometimes efficiently extract from a trained PQC a ``shadow model'' which can be used for efficient classical inference. Given this, it would be interesting to understand the extent to which one can train classical shadow models directly from data.

\section*{Acknowledgments}
RS is grateful for helpful conversations with Alex Nietner, Christa Zoufal and Yanting Teng.  This work is supported by the Government of Spain (Severo Ochoa CEX2019-000910-S, FUNQIP and European Union NextGenerationEU PRTR-C17.I1), Fundació Cellex, Fundació Mir-Puig, Generalitat de Catalunya (CERCA program) and European Union (PASQuanS2.1, 101113690). ERA is a fellow of Eurecat's "Vicente López" PhD grant program. ERA would also like to give special thanks to Jens Eisert for inviting him to participate in his group and also thank Adan Garriga for helping him all the way with this stay. SJ thanks the BMBK (EniQma) and the Einstein Foundation (Einstein Research Unit on Quantum Devices) for their support.
EGF is funded by the Einstein Foundation (Einstein Research Unit on Quantum Devices). JJM is funded by QuantERA (HQCC). 
JE is funded by the QuantERA (HQCC), the Einstein Foundation (Einstein Research Unit on Quantum Devices), the Munich Quantum Valley (K-8), Berlin Quantum, the QuantERA (HQCC), the MATH+ cluster of excellence, the BMWK (EniQma), and the BMBF (Hybrid++). 
We acknowledge support by the Open Access Publication Fund of Freie Universität Berlin.

\appendix

\section{Construction of the frequency set $\tilde{\Omega}_{\mathcal{D}}$}\label{app:freq_construction}

We describe here the way in which the frequency set $\tilde{\Omega}_\mathcal{D}$ of a PQC model is constructed from the data-encoding strategy $\mathcal{D}$. We follow closely the presentation in Ref.~\cite{RFF}, and start by noting that
\begin{equation}\label{eq:cartesian_freqs}
\tilde{\Omega}_\mathcal{D} = \tilde{\Omega}^{(1)}_\mathcal{D}\times \ldots \times \tilde{\Omega}^{(d)}_\mathcal{D}
\end{equation}
where $\tilde{\Omega}^{(j)}_\mathcal{D}\subseteq\mathbb{R}$ depends only on $\mathcal{D}^{(j)}$. We can therefore focus on the construction of $\tilde{\Omega}^{(j)}_\mathcal{D}$ for a single co-ordinate. In light of this, let us drop some coordinate-indicating superscripts for ease of presentation. In particular, let us write $\mathcal{D}^{(j)} = \{H_k\,|\, k \in [L_j] \}$, where we have dropped the coordinate-indicating superscripts from the Hamiltonians. We then use $\lambda^{i}_{k}$ to denote the $i$'th eigenvalue of $H_k$, and $N_{k}$ to denote the number of eigenvalues of $H_k$. We also introduce the multi-index $\vec{i} = (i_1,\ldots,i_{L_j})$, with $i_k\in [N_k]$, which allows us to define the sum of the eigenvalues indexed by $\vec{i}$, one from each Hamiltonian, as
\begin{equation}
\Lambda_{\vec{i}} = \lambda^{i_1}_1 +\ldots + \lambda^{i_{L_j}}_{L_j}.
\end{equation}
With this setup, we then have that the frequency set $\tilde{\Omega}^{(j)}_\mathcal{D}$ is given by the set of all differences of all possible sums of eigenvalues, i.e., 
\begin{equation}
\tilde{\Omega}^{(j)}_\mathcal{D} = \left\{\Lambda_{\vec{i}} - \Lambda_{\vec{j}}\,|\, \vec{i},\vec{j}\right\},
\end{equation}
and as mentioned before, the total frequency set is given by Eq.~\eqref{eq:cartesian_freqs}. There is a convenient graphical way to understand this construction, which is illustrated in Figure.~\ref{fig:tree_construction}. Essentially, one notes that, in order to construct $\tilde{\Omega}^{(j)}_\mathcal{D}$ one can consider a tree, with depth equal to the number of data-encoding gates, whose leaves contain the eigenvalue sums $\Lambda_{\vec{i}}$. The frequency set is then given by all possible pairwise differences between leaves.

\begin{figure}[h!]
\begin{center}
\includegraphics[scale=.6]{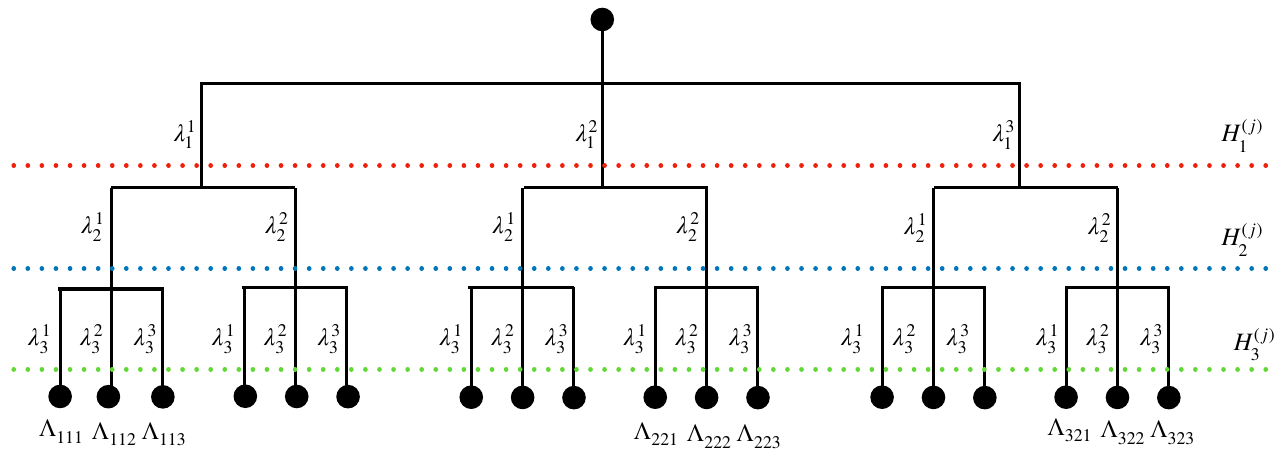}
\caption{Construction of the frequency set $\tilde{\Omega}^{(j)}_\mathcal{D}$ from the data-encoding strategy $\mathcal{D}^{(j)}$.}
\label{fig:tree_construction}
\end{center}
\end{figure}
\newpage
\section{Proof of Theorem~\ref{thm:efficiency_RFF}}\label{app:proof_main}

We start by noting that Theorem~\ref{thm:efficiency_RFF} in the main text follows as an immediate corollary of the following Theorem:

\begin{theorem}[RFF vs.~variational QML -- alternative form]\label{thm:efficiency_RFF_appendix} Let $R$ be the risk associated with a regression problem $P\sim\mathcal{X}\times\mathbb{R}$. Assume the following:
\begin{enumerate}
\item $\|f^*_{(\Theta,\mathcal{D},O)}\|_{K_\mathcal{D}} \leq C$,
\item $|y|\leq b$ almost surely when $(x,y)\sim P$, for some $b>0$.
\end{enumerate}
Additionally, define
\begin{align}
n_0 &\coloneqq \max\left\{4\|T_{K_\mathcal{D}}\|^2, \left(528\log\frac{1112\sqrt{2}}{\delta}\right)^2\right\},\\
c_0&\coloneqq 36\left(3 + \frac{2}{\|T_{K_\mathcal{D}}\|}\right), \\
c_1&\coloneqq 8\sqrt{2}(4b + \frac{5}{\sqrt{2}}C + 2\sqrt{2C}).
\end{align}
Then, let $\delta\in (0,1]$, let $n\geq n_0$, set $\lambda_n = 1/\sqrt{n}$, and let $\hat{f}_{M_n,\lambda_n}$ be the output of $\lambda_n$-regularized linear regression with respect to the feature map
\begin{equation}
\phi_{M_n}(x) = \frac{1}{\sqrt{M_n}}\big(\psi(x,\nu_1),\ldots,\psi(x,\nu_{M_n})\big)
\end{equation}
constructed from the integral representation of $K_\mathcal{D}$ by sampling $M_n$ elements from $\pi$. Then, 
\begin{equation}
M_n \geq c_0\sqrt{n}\log\frac{108\sqrt{n}}{\delta}
\end{equation}
is enough to guarantee, with probability at least $1-\delta$, that if $R(\hat{f}_{M_n,\lambda_n}) \geq R(f^*_{(\Theta,\mathcal{D},O)})$, then
\begin{equation}
R(\hat{f}_{M_n,\lambda_n}) - R(f^*_{(\Theta,\mathcal{D},O)}) \leq \frac{c_1\log^2\frac{1}{\delta}}{\sqrt{n}}.
\end{equation}
\end{theorem}
Next we note that Theorem~\ref{thm:efficiency_RFF_appendix} above -- from which we derive Theorem~\ref{thm:efficiency_RFF} in the main text as an immediate corollary -- is essentially a straightforward application of the generalization bound given as Theorem 1 in Ref.~\cite{genRFF}. As such, we start our proof of Theorem~\ref{thm:efficiency_RFF_appendix} with a presentation of this result. To this end, we require first a few definitions. Firstly, given a kernel $K$, with associated RKHS $(\mathcal{H}_K,\langle\cdot,\cdot\rangle_K)$, we define $\mathcal{H}^C_K = \{f\in\mathcal{H}_K\,|\, \|f\|_K\leq C\}$ as the subset of functions in $\mathcal{H}_K$ with RKHS norm bounded by $C$. We define $\mathcal{F}^C_\mathcal{D}$ and $\mathcal{F}^C_{(\Theta,\mathcal{D},O)}$ analogously. Additionally, given a regression problem $P$ with associated risk $R$, we then define 
\begin{equation}
f^*_{\mathcal{H}^C_K} = \argmin_{f\in \mathcal{H}^C_K}\left[R(f)\right],
\end{equation}
as the optimal function for $P$ in $\mathcal{H}^C_K$. Finally, recall that we denote by $T_K$ the kernel integral operator associated with the kernel $K$ (see Definition~\ref{def:kern_int_op}).  With this in hand, we can state a slightly reformulated version of the RFF generalization bound proven in Ref.~\cite{genRFF} (which in turn build on the earlier results of Ref.~\cite{RFFOld}).

\begin{theorem}[Theorem 1 from~\cite{genRFF}]\label{thm:gen_bound} Assume a regression problem $P\sim\mathcal{X}\times\mathbb{R}$. Let $K:\mathcal{X}\times\mathcal{X}\rightarrow \mathbb{R}$ be a kernel, and let $\mathcal{H}^C_K$ be the subset of the RKHS $\mathcal{H}_K$ consisting of functions with RKHS-norm upper bounded by some constant $C$. Assume the following:
\begin{enumerate}
\item K has an integral representation 
\begin{equation}
K(x,x') = \int_{\Phi}\psi(x,\nu)\psi(x',\nu)\,\mathrm{d}\pi(\nu).
\end{equation}
\item The function $\psi$ is continuous in both variables
 and satisfies $|\psi(x,\nu)|\leq \kappa$ almost surely, for some $\kappa \in [1,\infty)$.
\item $|y|\leq b$ almost surely when $(x,y)\sim P$, for some $b > 0$.
\end{enumerate}
Additionally, define
\begin{align}
\overline{B} &\coloneqq 2b + 2\kappa\max\{1,\|f^*_{\mathcal{H}^C_K}\|_K\},\\
\overline{\sigma}&\coloneqq2b + 2\kappa\sqrt{\max\{1,\|f^*_{\mathcal{H}^C_K}\|_{K}\}},
\end{align}
and
\begin{align}
n_0 &\coloneqq \max\left\{4\|T_K\|^2 ,\left(264\kappa^2\log\frac{556\kappa^3}{\delta}\right)^2\right\},\\
c_0&\coloneqq9\left(3 + 4\kappa^2 + \frac{4\kappa^2}{\|T_K\|} + \frac{\kappa^4}{4}\right), \\
c_1&\coloneqq8\big(\overline{B}\kappa + \overline{\sigma}\kappa + \max\{1,\|f^*_{\mathcal{H}^C_K}\|_K\}\big).
\end{align}
Then, let $\delta\in (0,1]$,  $n\geq n_0$, assume $\lambda_n = 1/\sqrt{n}$, and let $\hat{f}_{M_n,\lambda_n}$ be the output of $\lambda_n$-regularized linear regression with respect to the feature map
\begin{equation}
\phi_{M_n}(x) = \frac{1}{\sqrt{M_n}}\big(\psi(x,\nu_1),\ldots\psi(x,\nu_M)\big]
\end{equation}
constructed from the integral representation of $K$ by sampling $M_n$ elements from $\pi$. Then, 
\begin{equation}\label{Me}
M_n \geq c_0\sqrt{n}\log\frac{108\kappa^2\sqrt{n}}{\delta}
\end{equation}
is enough to guarantee, with probability at least $1-\delta$, that
\begin{equation}
R(\hat{f}_{M_n,\lambda_n}) - R(f^*_{\mathcal{H}^C_K}) \leq \frac{c_1\log^2\frac{1}{\delta}}{\sqrt{n}}.
\end{equation}
\end{theorem}
This statement demonstrates, that under
reasonable assumptions, the estimator that is obtained with a number of random features proportional to 
$\mathcal{O}(\sqrt{n} \log n)$ achieves a 
$\mathcal{O}(1/\sqrt{n})$  learning error. We would now like to prove Theorem~\ref{thm:efficiency_RFF_appendix} by applying Theorem~\ref{thm:gen_bound} to the classical PQC-kernel $K_\mathcal{D}$. To do this, we require the following Lemma:

\begin{lemma}[Function set inclusions]\label{lem:function_inclusions} For any constant $C$ one has that
\begin{align}
\mathcal{F}^C_{(\Theta,\mathcal{D},O)} \subseteq &\,\mathcal{F}^{C}_{\mathcal{D}} \subseteq \mathcal{F}_{\mathcal{D}} \subseteq \mathcal{H}_{K_\mathcal{D}}\nonumber\\
&\,\,\rsubseteq\\ 
&\mathcal{H}^{C}_{K_\mathcal{D}}.\nonumber
\end{align}
\end{lemma}

\begin{proof}[Proof of Lemma~\ref{lem:function_inclusions}]
The inclusions $\mathcal{F}^C_{(\Theta,\mathcal{D},O)} \subseteq \mathcal{F}^{C}_{\mathcal{D}} \subseteq \mathcal{F}_{\mathcal{D}}$ are immediate by definition. To prove $\mathcal{F}_{\mathcal{D}} \subseteq \mathcal{H}_{K_\mathcal{D}}$ we consider any $f_v\in\mathcal{F}_\mathcal{D}$ and let $\tilde{\mathcal{F}}$ be the image of $\mathcal{X}$ under $\phi_\mathcal{D}$. We can then write $v = v_1 + v_2$ where $v_1\in\tilde{\mathcal{F}}$ and $v_2$ lies in the remainder. We then have that
\begin{align}
f_v(x) &=  \langle v,\phi_\mathcal{D}(x) \rangle \\
\nonumber
&= \langle v_1,\phi_\mathcal{D}(x) \rangle + \langle v_2,\phi_\mathcal{D}(x) \rangle \\
\nonumber
&= \langle v_1,\phi_\mathcal{D}(x) \rangle .
\nonumber
\end{align}
By definition -- i.e., the fact that $v_1\in\tilde{\mathcal{F}}$ -- we know that we can write
\begin{equation}
v_1 = \sum_{j}\gamma_j\phi_\mathcal{D}(x_j)
\end{equation}
and, therefore, 
we see that 
\begin{align}
f_v(x) &= \sum_j\gamma_j\langle \phi_\mathcal{D}(x_j),\phi_\mathcal{D}(x)\rangle_\mathcal{F} \\
&= \sum_j\gamma_j K_\mathcal{D}(x_j,x),
\nonumber
\end{align}
i.e., $f_v$ is indeed in the RKHS $\mathcal{H}_{K_\mathcal{D}}$. Finally, the inclusion $\mathcal{F}^C_\mathcal{D}\subseteq \mathcal{H}^{C}_{K_\mathcal{D}}$ now follows easily, as $\mathcal{F}^C_{\mathcal{D}} \subseteq \mathcal{F}_{\mathcal{D}}\subseteq\mathcal{H}_{K_\mathcal{D}}$ yields
\begin{equation}
f\in \mathcal{F}^C_{\mathcal{D}} \implies f\in\mathcal{H}_{K_\mathcal{D}},
\end{equation}
and by definition
\begin{equation}
f\in \mathcal{F}^C_{\mathcal{D}} \implies \|f\|_{K_\mathcal{D}} \leq C
\end{equation}
which together means that
\begin{equation}
f\in \mathcal{F}^C_{\mathcal{D}} \implies f\in\mathcal{H}^{C}_{K_\mathcal{D}}.
\end{equation}
\end{proof}

With this in hand, we can now prove Theorem~\ref{thm:efficiency_RFF_appendix}. 

\begin{proof}[Proof of Theorem~\ref{thm:efficiency_RFF_appendix}] We start by recalling that, as shown in Section~\ref{ss:implementation}, for any reweighting vector $\doubleyou$ the reweighted PQC-kernel $K_{(\mathcal{D},\doubleyou)}$ has the integral representation
\begin{align}
K_{(\mathcal{D},\doubleyou)}(x,x') = \frac{1}{2\pi}\int_{\mathcal{X}}\int_0^{2\pi}\sqrt{2}\cos(\langle\omega, x\rangle + \gamma)\sqrt{2}\cos(\langle\omega, x'\rangle + \gamma)q_{(\mathcal{D},\doubleyou)}(\omega)\,\,\mathrm{d}\gamma \mathrm{d}\nu, 
\end{align}
where 
\begin{align}
q_{(\mathcal{D},\doubleyou)}(\omega) =\sum_{i = 0}^{|\Omega^+_\mathcal{D}|}\frac{\doubleyou_i^2}{\|\doubleyou\|^2_2}\delta(\omega - \omega_i).
\end{align}
As a result, for any reweighting, including $\doubleyou= (1,\ldots,1)$, the kernel $K_{(\mathcal{D},\doubleyou)}$ satisfies assumption (1) of Theorem~\ref{thm:gen_bound} with 
\begin{equation}
\psi(x,\nu) = \sqrt{2}\cos(\langle \omega, x\rangle + \gamma).
\end{equation}
Given this, we note that $\psi$ is continuous in both variables and that $|\psi(x,\nu)|\leq \sqrt{2}$ for all $x,\nu$ -- i.e.,  for any kernel $K_{(\mathcal{D},\doubleyou)}$, assumption~(2) of Theorem~\ref{thm:gen_bound} is satisified with $\kappa = \sqrt{2}$.

Next, set the $C$ appearing in Theorem~\ref{thm:gen_bound} to the constant $C$ appearing in assumption (1) of Theorem~\ref{thm:efficiency_RFF_appendix}. More specifically, we apply Theorem~\ref{thm:gen_bound} to the subset $\mathcal{H}_{K_\mathcal{D}}^C$, where $C$ is an upper bound on the RKHS norm of the optimal function for $P$ in $\mathcal{F}_{(\Theta,\mathcal{D},O)}$ -- i.e., $\|f^*_{(\Theta,\mathcal{D},O)}\|_{K_\mathcal{D}}\leq C$. Doing this we obtain, via Theorem~\ref{thm:gen_bound} and the fact that $\kappa = \sqrt{2}$, that provided all the conditions of Theorem~\ref{thm:efficiency_RFF_appendix} are satisfied, then
\begin{equation}\label{eq:intermediate_step}
R(\hat{f}_{M_n,\lambda_n}) - R(f^*_{\mathcal{H}^C_{K_\mathcal{D}}}) \leq \frac{c_1\log^2\frac{1}{\delta}}{\sqrt{n}}.
\end{equation}
To achieve the statement of Theorem~\ref{thm:efficiency_RFF} we then use the assumption that $\|f^*_{(\Theta,\mathcal{D},O)}\|_{K_\mathcal{D}}\leq C$. More specifically, via Lemma~\ref{lem:function_inclusions} this assumption implies that  $f^*_{(\Theta,\mathcal{D},O)}\in \mathcal{H}^C_{K_\mathcal{D}}$, which together with the definition of $f^*_{\mathcal{H}^C_{K_\mathcal{D}}}$ as the \textit{optimal} function in $\mathcal{H}^C_{K_\mathcal{D}}$, allows us to conclude that 
\begin{equation}\label{eq:pqc_vs_rff_loss}
R(f^*_{(\Theta,\mathcal{D},O)}) \geq R(f^*_{\mathcal{H}^C_{K_\mathcal{D}}}).
\end{equation}
This then implies
\begin{align}
R(\hat{f}_{M_n,\lambda_n}) - R(f^*_{(\Theta,\mathcal{D},O)}) &\leq R(\hat{f}_{M_n,\lambda_n}) - R(f^*_{\mathcal{H}^C_{K_\mathcal{D}}}) \\
&\leq \frac{c_1\log^2\frac{1}{\delta}}{\sqrt{n}}\hspace{8em}\text{[via Eq.~\eqref{eq:intermediate_step}]} 
\nonumber
\end{align}
as per the statement of Theorem~\ref{thm:efficiency_RFF_appendix}.
\end{proof}
As already mentioned, Theorem~\ref{thm:efficiency_RFF} in the main text then follows as an immediate corollary of Theorem~\ref{thm:efficiency_RFF_appendix}.

\section{Proof of Lemma~\ref{lem:operator_norm}}\label{app:kern_int_proof}

\begin{proof}[Proof of Lemma~\ref{lem:operator_norm}]
As discussed in Ref.~\cite{rosasco2010learning}, the kernel integral operator is self-adjoint. In light of this, we know that $ \|T_{K_{(\mathcal{D},\doubleyou)}}\| =\rho(T_{K_{(\mathcal{D},\doubleyou)}})$, where $\rho(T_{K_{(\mathcal{D},\doubleyou)}})$ denotes the \textit{spectral radius} of $\rho(T_{K_{(\mathcal{D},\doubleyou)}})$. As such, we focus on determining the spectrum of $\rho(T_{K_{(\mathcal{D},\doubleyou)}})$. To this end, note that under assumption (a) of the lemma statement we have that 
\begin{align}
(T_{K_{(\mathcal{D},\doubleyou)}}g)(x) &= \int_\mathcal{X}K_{(\mathcal{D},\doubleyou)}(x,x')g(x')dP_\mathcal{X}(x') \\
\nonumber
&= \frac{1}{(2\pi)^d}\int_\mathcal{X}K_{(\mathcal{D},\doubleyou)}(x,x')g(x')\,\mathrm{d}x' \qquad\qquad\text{[via assumption (a)]}
\end{align}
with 
\begin{align}
    K_{(\mathcal{D},\doubleyou)}(x,x') &= \frac{1}{\|\doubleyou\|^2_2}\left(\doubleyou^2_0 + \sum_{i=1}^{ |\Omega^{+}_\mathcal{D}|}\doubleyou_i^2\cos(\langle\omega_i, (x-x')\rangle)\right),\\
    \nonumber
\end{align}
where $\omega_0 = (0,\ldots,0)$. We now use Assumption (b) -- i.e., that $\Omega_{\mathcal{D}^+}\subset\mathbb{Z}^d$ -- to show that for any $\omega\in\mathbb{Z}^d$, the function $g(x') = \cos(\langle\omega, x'\rangle)$ is an eigenfunction of $T_{K_{(\mathcal{D},\doubleyou)}}$. Specifically,  using the following notation
\begin{equation}
\delta(\omega\pm\nu) \coloneqq \begin{cases} 1 \hspace{1em}\text{if } (\omega=\nu) \lor (\omega =-\nu),\\
0\hspace{1em}\text{else},
\end{cases}
\end{equation}
and defining $\doubleyou_\omega$ to be the weight associated with $\omega\in\Omega_\mathcal{D}$, we 
have that
\begin{align}
    (T_{K_{(\mathcal{D},\doubleyou)}}g)(x) &= \frac{1}{(2\pi)^d}\int_\mathcal{X} \left[\frac{1}{\|\doubleyou\|^2_2}\left(\doubleyou^2_0 + \sum_{i=1}^{ |\Omega^{+}_\mathcal{D}|}\doubleyou_i^2\cos(\langle\omega_i, (x-x')\rangle)\right)\right] \cos(\langle\omega, x'\rangle)\,\mathrm{d}x' \\
    \nonumber
    &=\frac{1}{(2\pi)^d\|\doubleyou\|_2^2}\left[\int_\mathcal{X} w^2_0\cos(\langle\omega, x'\rangle)\,\mathrm{d}x' + \int_\mathcal{X}\sum_{i=1}^{ |\Omega^{+}_\mathcal{D}|}\doubleyou_i^2\cos(\langle\omega_i, (x-x')\rangle)\cos(\langle\omega, x'\rangle)\,\mathrm{d}x'\right]
    \\
    \nonumber
    &=\frac{1}{(2\pi)^d\|\doubleyou\|_2^2}\int_\mathcal{X}\sum_{i=1}^{ |\Omega^{+}_\mathcal{D}|}\doubleyou_i^2\cos(\langle\omega_i,(x-x')\rangle) \cos(\langle\omega, x'\rangle)\,\mathrm{d}x'\hspace{3em}\text{[via assumption (b)]} \\
    \nonumber
    &=\frac{1}{(2\pi)^d\|\doubleyou\|_2^2}\sum_{i=1}^{ |\Omega^{+}_\mathcal{D}|}\doubleyou_i^2\int_\mathcal{X}\cos(\langle\omega_i,(x-x')\rangle) \cos(\langle\omega, x'\rangle)\,\mathrm{d}x' \\
    \nonumber
    &= \frac{1}{(2\pi)^d\|\doubleyou\|^2_2}\sum_{i =1}^{|\Omega^{+}_\mathcal{D}|}\doubleyou_i^2 \frac{(2\pi)^d}{2}\cos(\langle\omega_i,x\rangle)\delta(\omega_i\pm\omega) \hspace{7.5em}\text{[via assumption (b)]}\\
    \nonumber
    &= \begin{cases} 
    \frac{\doubleyou_\omega^2}{2\|\doubleyou\|^2_2}\cos(\langle\omega,x\rangle) \hspace{1em}\text{if }\omega\in\Omega^+_\mathcal{D},\\
    \nonumber
    \frac{\doubleyou_{-\omega}^2}{2\|\doubleyou\|^2_2}\cos(\langle\omega,x\rangle) \hspace{1em}\text{if }\omega\in-\Omega^+_\mathcal{D},\\
    0 \hspace{6.8em}\text{else}.
    \nonumber
    \end{cases} 
\end{align}
A similar calculation shows that, for all $\omega\in\mathbb{Z}^d$,
\begin{equation}
\left(T_{K_{(\mathcal{D},\doubleyou)}}\sin(\langle\omega,x'\rangle)\right)(x)    = \begin{cases} 
    \frac{\doubleyou_\omega^2}{2\|\doubleyou\|^2_2}\sin(\langle\omega,x\rangle) \hspace{2em}\text{if }\omega\in\Omega^+_\mathcal{D},\\
    \frac{\doubleyou_{-\omega}^2}{2\|\doubleyou\|^2_2}\sin(\langle\omega,x\rangle) \hspace{1em}\text{if }\omega\in-\Omega^+_\mathcal{D},\\
    0 \hspace{7.2em}\text{else}.
    \end{cases} 
\end{equation}
As such, we have that all functions in the set $\{\sin\langle\omega,x\rangle)\,|\,\omega\in\mathbb{Z}^d\}\cup \{\cos\langle\omega,x\rangle)\,|\,\omega\in\mathbb{Z}^d\}$ are eigenfunctions of $T_{K_\mathcal{D}}$. However, as this set is a basis for $L^2(\mathcal{X},P_\mathcal{X})$ -- in the relevant case where $P_\mathcal{X}$ is the uniform distribution -- we can conclude that
\begin{align}
  \|T_{K_{(\mathcal{D},\doubleyou)}}\| &=\rho\left(T_{K_{(\mathcal{D},\doubleyou)}}\right)\\
  \nonumber
&= \max_{\omega \in \Omega_\mathcal{D}} \left\{\frac{1}{2}\frac{\doubleyou_\omega^2}{\|\doubleyou\|^2_2} \right\}\\
\nonumber
  &=\max_{\omega \in \Omega_\mathcal{D}}\left\{\frac{1}{2}p_{(\mathcal{D},\doubleyou)}(\omega)\right\}.
\end{align}
\end{proof}
As an aside, it is interesting to note that the minimization of the norm $\|T_{K_{(\mathcal{D},\doubleyou)}}\|$ subject to the constraint $\|\doubleyou\|^2_2\leq c_0$ 
for some $c_0>0$ can be captured in terms of a 
convex semi-definite problem. 
This problem can be written as
\begin{eqnarray}
\text{minimize } &c&\\
\text{subject to }
\frac{1}{2}
\frac{\doubleyou_\omega^2}{\|\doubleyou\|^2_2}&\leq& c,\\
\|\doubleyou\|^2_2&\leq &c_0,
\nonumber
\end{eqnarray}
which is easily seen to be equivalent with
%
% Intermediate step
%
%\begin{eqnarray}
%\text{minimize } c\\
%\text{subject to }
%\left[
%\begin{array}{cc}
%2c & |\doubleyou_\omega|\\
%|\doubleyou_\omega| & d
%\end{array}
%\right]\geq 0
%\text{for all }\omega
%,\\
%d\leq c_0,\\
%d = \sum_\omega  %\doubleyou_\omega^2.
%\end{eqnarray}
\begin{eqnarray}
\text{minimize } &c&,\\
\text{subject to }
\left[
\begin{array}{cc}
2c & |\doubleyou_\omega|\\
|\doubleyou_\omega| & d
\end{array}
\right]&\geq& 0\,
\text{for all }\omega
,\\
\nonumber
d&\leq& c_0,\\
\nonumber
\left[
\begin{array}{cc}
d & \doubleyou\\
\doubleyou^T & I
\end{array}
\right]&\geq &0,
\nonumber
\end{eqnarray}
by making use of Schur complements.

\section{Proof of Lemma~\ref{lem:lower_bound}}\label{app:lower}

\begin{proof}[Proof of Lemma~\ref{lem:lower_bound}] As $g_{\vec{\omega}}(x)$ is the output of the RFF procedure in which frequencies $\vec{\omega} = (\omega_1,\ldots,\omega_M)$ were drawn, we know that $g_{\vec{\omega}}$ can be written as
\begin{equation}
g_{\vec{\omega}}(x) = \sum_{\omega\in\tilde{\Omega}_\mathcal{D}}\hat{g}_{\vec{\omega}}(\omega)e^{i\langle\omega, x\rangle},
\end{equation}
where $\hat{g}_{\vec{\omega}}(\omega) = 0$ for all $\omega \notin \{\omega_i\}|_{i=1}^M$. Again we abuse notation and use $\hat{g}_{\vec{\omega}}$ to denote the vector with entries $\hat{g}_{\vec{\omega}}(\omega)$. Now, given some vector $\vec{\omega} = (\omega_1,\ldots,\omega_M) \in \Omega^M_\mathcal{D}$, we define the sets $\Omega_{\vec{\omega}}\coloneqq\{\omega_1,\ldots,\omega_M\}\subseteq \Omega_\mathcal{D}$ and $\tilde{\Omega}_{\vec{\omega}}\coloneqq \Omega_{\vec{\omega}}\cup(-\Omega_{\vec{\omega}})$. Given some $\hat{f}^*$, we then define the vectors
\begin{align}
\hat{f}^*_{\vec{\omega}}(\omega) &= \begin{cases} \hat{f}^*(\omega) \text{ if } \omega\in\tilde{\Omega}_{\vec{\omega}} ,\\
0 \text{ else},
\end{cases}\\
\hat{f}^*_{/\vec{\omega}}(\omega) &= \begin{cases} 0 \text{ if } \omega\in\tilde{\Omega}_{\vec{\omega}} ,\\
\hat{f}^*(\omega) \text{ else}.
\end{cases}
\end{align}
Note that with these definitions, $\hat{f}^* = \hat{f}^*_{\vec{\omega}} + \hat{f}^*_{/\vec{\omega}}$. Using this, we have that
\begin{align}
\|\hat{f}^* - \hat{g}_{\vec{\omega}}\|_2^2 &= \|\hat{f}^*_{\vec{\omega}} + \hat{f}^*_{/\vec{\omega}} - \hat{g}_{\vec{\omega}}\|_2^2 \label{eq:two_norm_split_1}\\
\nonumber
&= \|\hat{f}^*_{/\vec{\omega}}\|^2_2 + \|\hat{f}^*_{\vec{\omega}} - \hat{g}_{\vec{\omega}}\|_2^2\\
&\geq \|\hat{f}^*_{/\vec{\omega}}\|^2_2 \label{eq:two_norm_split_2}\\
&= \|\hat{f}^*\|^2_2 - \|\hat{f}^*_{\vec{\omega}}\|_2^2. \label{eq:two_norm_split}
\end{align}
Using this expression, we can then lower-bound $\hat{\epsilon}$, the expected $L^2$-norm of the difference between the optimal PQC function and the output of the RFF procedure, recalling that $\xi(\vec{\omega})$ is the probability of sampling the vector of frequencies $\vec{\omega}$,
\begin{align}
\hat{\epsilon} &=\sum_{\vec{\omega}\in \Omega^M_\mathcal{D}}\|f^* - g_{\vec{\omega}}\|_2^2\,\xi(\vec{\omega})\label{eq:before_pars}\\
&= (2\pi)^d\sum_{\vec{\omega}\in \Omega^M_\mathcal{D}}\|\hat{f}^* - \hat{g}_{\vec{\omega}}\|_2^2\xi(\vec{\omega})\hspace{5em} \text{ [via Parseval's identity]}\label{eq:after_pars}\\
\nonumber
&\geq (2\pi)^d\sum_{\vec{\omega}\in \Omega^M_\mathcal{D}}\left[\|\hat{f}^*\|^2_2 - \|\hat{f}^*_{\vec{\omega}}\|_2^2\right]\xi(\vec{\omega})\hspace{3em} \text{ [via Eq.~\eqref{eq:two_norm_split}]} \\
&= (2\pi)^d\|\hat{f}^*\|^2_2 - (2\pi)^d\sum_{\vec{\omega}\in \Omega^M_\mathcal{D}}\|\hat{f}^*_{\vec{\omega}}\|_2^2\,\xi(\vec{\omega})\label{eq:sub_into}.
\end{align}
Note that we used the assumption of integer-valued frequency vectors -- and therefore orthogonal components of the feature map -- when invoking Parseval's identity to move from line~\eqref{eq:before_pars} to \eqref{eq:after_pars}. Using the short-hand notation $0$ to denote the frequency vector $(0,\ldots,0)$, we can now analyze the final term as 
\begin{align}
\nonumber
\sum_{\vec{\omega} \in \Omega^M_\mathcal{D}} \|\hat{f}^*_{\vec{\omega}}\|_2^2 \; \xi(\vec{\omega})
&= \sum_{\substack{\vec{\omega} \in \Omega^M_\mathcal{D}\\ 
\nonumber
0\notin\Omega_{\vec{\omega}}}} \sum_{i=1}^{M} \left(|\hat{f}^*(-\omega_i)|^2+|\hat{f}^*(\omega_i)|^2\right) \xi(\vec{\omega})\\
&\qquad\qquad+\sum_{\substack{\vec{\omega} \in \Omega^M_\mathcal{D}\\ 
\nonumber
0\in\Omega_{\vec{\omega}}}}\left[ \sum_{i=1}^{M-1} \left(|\hat{f}^*(-\omega_i)|^2+|\hat{f}^*(\omega_i)|^2\right) + |\hat{f}^*(0)|^2 \right] \xi(\vec{\omega})\\
\nonumber
&\leq \sum_{\substack{\vec{\omega} \in \Omega^M_\mathcal{D}\\ 
\nonumber
0\notin\Omega_{\vec{\omega}}}} \sum_{i=1}^{M} 2|\hat{f}^*(\omega_i)|^2 \xi(\vec{\omega})\\
\nonumber
&\qquad\qquad+ \sum_{\substack{\vec{\omega} \in \Omega^M_\mathcal{D}\\ 
\nonumber
0\in\Omega_{\vec{\omega}}}}\left[ \sum_{i=1}^{M-1} 2|\hat{f}^*(\omega_i)|^2 + 2|\hat{f}^*(0)|^2 \right] \xi(\vec{\omega})\\
\nonumber
&=  2\sum_{i=1}^{M} \sum_{\omega_1 \in \Omega_\mathcal{D}}\cdot\cdot\cdot \sum_{\omega_{M} \in \Omega_\mathcal{D}} |\hat{f}^*(\omega_i)|^2 p(\omega_1)\cdot\cdot\cdot p(\omega_{M}) \\
\nonumber
&= 2\sum_{i=1}^{M} \sum_{\omega_i \in \Omega_\mathcal{D}} |\hat{f}^*(\omega_i)|^2 p(\omega_i)\\
&= 2M \sum_{\nu \in \Omega_\mathcal{D}} |\hat{f}^*(\nu)|^2p(\nu).\label{eq:final_step}
\end{align}
Substituting Eq.~\eqref{eq:final_step} into Eq.~\eqref{eq:sub_into} then gives the statement of the Lemma.
\end{proof}
\newpage
\printbibliography

\end{document}